\documentclass[11pt,twoside]{article}
\usepackage[T1]{fontenc} 

\usepackage{amsthm, amsmath, amssymb, mathrsfs, upgreek, enumerate, mathtools,
  comment, natbib, subfigure, graphicx, mathabx, booktabs, threeparttable,
tikz, relsize, exscale, epstopdf, multirow, dsfont}
\usepackage[headheight=110pt,margin=1.25in]{geometry}
\usepackage[section]{placeins} 
\renewcommand{\(}{\left(}
\renewcommand{\)}{\right)}
\renewcommand{\[}{\left[}
\renewcommand{\]}{\right]}
\usepackage{fancyhdr, setspace} 
\usepackage{tcolorbox}
\usepackage[nolists,noheads,nomakers]{endfloat}

\usepackage{hyperref}
\usepackage{subfigure}
\usepackage{marginnote}
\hypersetup{
    colorlinks=true,
    citecolor=blue,
    filecolor=black,
    linkcolor=red,
    urlcolor=blue,
    bookmarksopen=true,
    pdfstartview=FitH
}
\usepackage{CJK}
\usepackage[utf8]{inputenc}

\newtheorem*{theorem*}{Theorem}

\newtheorem{lemma}{Lemma}[section]

\theoremstyle{definition}

\makeatletter
\def\thm@space@setup{
  \thm@preskip
15pt \thm@postskip=15pt 
}
\makeatother

\renewcommand{\qed}{\hfill \mbox{\raggedright \rule{0.08in}{0.08in}}} 
\renewenvironment{proof}[1][\proofname]{{\noindent\sc#1. }}{\qed\vspace{15pt}} 

\title{\bf\sc Decision Making with Machine Learning and ROC Curves}

\author{Kai Feng\thanks{Department of Computer Science, Beihang University. fengkai@buaa.edu.cn} \and Han Hong\thanks{Department of Economics, Stanford University. doubleh@stanford.edu} \and
Ke Tang\thanks{Institute of Economics, School of Social Science, Tsinghua University. ketang@tsinghua.edu.cn}
\and Jingyuan
Wang\thanks{Department of Computer Science, Beihang University. jywang@buaa.edu.cn}
}

\begin{document}
\maketitle

\begin{abstract}
  \noindent {\sc Abstract.} The Receiver Operating Characteristic (ROC) curve is a representation of the statistical
information discovered in binary classification problems and is a key concept in
machine learning and data science. This paper studies the statistical properties
of ROC curves and its implication on model selection. We analyze the
implications of different models of incentive heterogeneity and information
asymmetry on the relation between human decisions and the ROC curves. Our
theoretical discussion is illustrated in the context of a large data set of
pregnancy outcomes and doctor diagnosis from the Pre-Pregnancy Checkups of
reproductive age couples in Henan Province
provided by the Chinese
Ministry of Health.
  \vspace{15pt}

\end{abstract}

\newpage

\section{Introduction}\label{introduction}
In the era of artificial intelligence, much attention has been drawn to the potential of machine learning in assisting human decision making. Among the recent headlines, Google used a deep learning algorithm to detect diabetic retinopathy in retinal fundus photographs (\cite{gulshan1}). \cite{long1} proposes a new deep learning algorithm, CC-Cruiser, which can be used to diagnose cataracts and provide treatment advice. In the economics literature, \cite{currie3}'s analysis of the decision making of physicians suggests the possibility of improvements that can both benefit patients and reduce medical expenses.

In this paper, we focus on the binary classification decision making problem, i.e. a choice between ``yes'' and ``no'', or $1$ and $0$. Binary classification is a foundational building block of statistical decision making in a variety of disciplines including machine learning, data science, and econometrics. It takes different forms and interpretations in different empirical settings: ``ill or healthy'' in a medical diagnosis by physicians, ``jail or release'' in a verdict by sitting Judges, and ``accept or reject'' by a collage admission committee.

A key concept that is used to evaluate the quality of binary classification and prediction is the Receiver Operating Characteristic (ROC) Curve, which essentially measures the diagnostic performance of a binary classifier as its cutoff threshold is varied. Since its first appearance in \cite{fisher1}, ROC curves and its variants are widely used in analyzing empirical data. However, its statistical properties do not appear to be well-understood. This paper serves to bridge such an important gap, and present statistical properties of ROC curves in the context of decision making under incentive heteroscedasticity and information asymmetry.

For a data set with labels $Y_i \in \{0,1\}$, and features $X_i,i=1,\ldots,n$, a learning algorithm makes predictions $\hat Y_i \in \{0,1\}$ for $Y_i$. An algorithm that generates a higher proportion of correct predictions, defined as the accuracy
\begin{eqnarray*}
  \text{ACC}=\frac{1}{n} \sum_{i=1}^n \mathds{1}\(\hat Y_i = Y_i\),
\end{eqnarray*}
naturally has more appeal.

However, accuracy itself is unlikely to be sufficient to characterize the quality of
prediction algorithms. In the National Free Pre-Pregnancy Check-ups (NFPC) data set that we studied as an application of this paper, the accuracy of doctors' diagnosis of high-risk pregnancy $\hat Y_i$ in predicting birth defects ($Y_i$) is about 80\%. Yet less than 5\% of the births are abnormal. A simple prediction of all pregnancy as low risk would result in an accuracy of 95\%. Nevertheless, the precision of doctor's diagnosis among the
abnormal births, or the true positive rate (TPR), is much higher than the naive
prediction of all normal births (where the TPR is by definition zero). For this
reason, the pair of TPR and false positive rate (FPR) are to be considered in tandem
in evaluating a given set of predictions:
{\begin{eqnarray*}
  \text{TPR} = \sum_{i=1}^n Y_i \hat Y_i/ \sum_{i=1}^n Y_i,\quad
  \text{FPR} = \sum_{i=1}^n \(1-Y_i\) \hat Y_i / \sum_{i=1}^n \(1-Y_i\).
\end{eqnarray*}}

In a statistical setting, a prediction method (based on either machine learning or a binary choice econometric model) typically generates a function of features that represents a sample estimate of the probability of the label taking value $1$
conditional on the features: $\hat p\(X_i\) \in \[0,1\]$.  The sample ROC
is a transformation of the estimated probability function $\hat p\(X_i\)$.
It is the collection of the set of all TPR/FPR pairs corresponding to decision rules
of the form of $\hat Y_i = \mathds{1}\(\hat p\(X_i\) > c\)$ when $c$ varies from $0$ to $1$. Essentially, ROC curves present the tradeoff between TPR and FPR for different cutoff thresholds. Figure \ref{figure 1} shows the typical shape of a ROC curve (the blue line),
e.g. resembling those reported in numerous science papers.
The red dot represents the aggregate TPR/FPR of human decision makers
(such as doctors diagnosing diseases) in the observed data that is typically
benchmarked by the ROC curves generated by machine learning algorithms.
\begin{figure}
  \begin{center}
    \caption{A typical ROC curve}\label{figure 1}
    \includegraphics[height=.32\textheight]{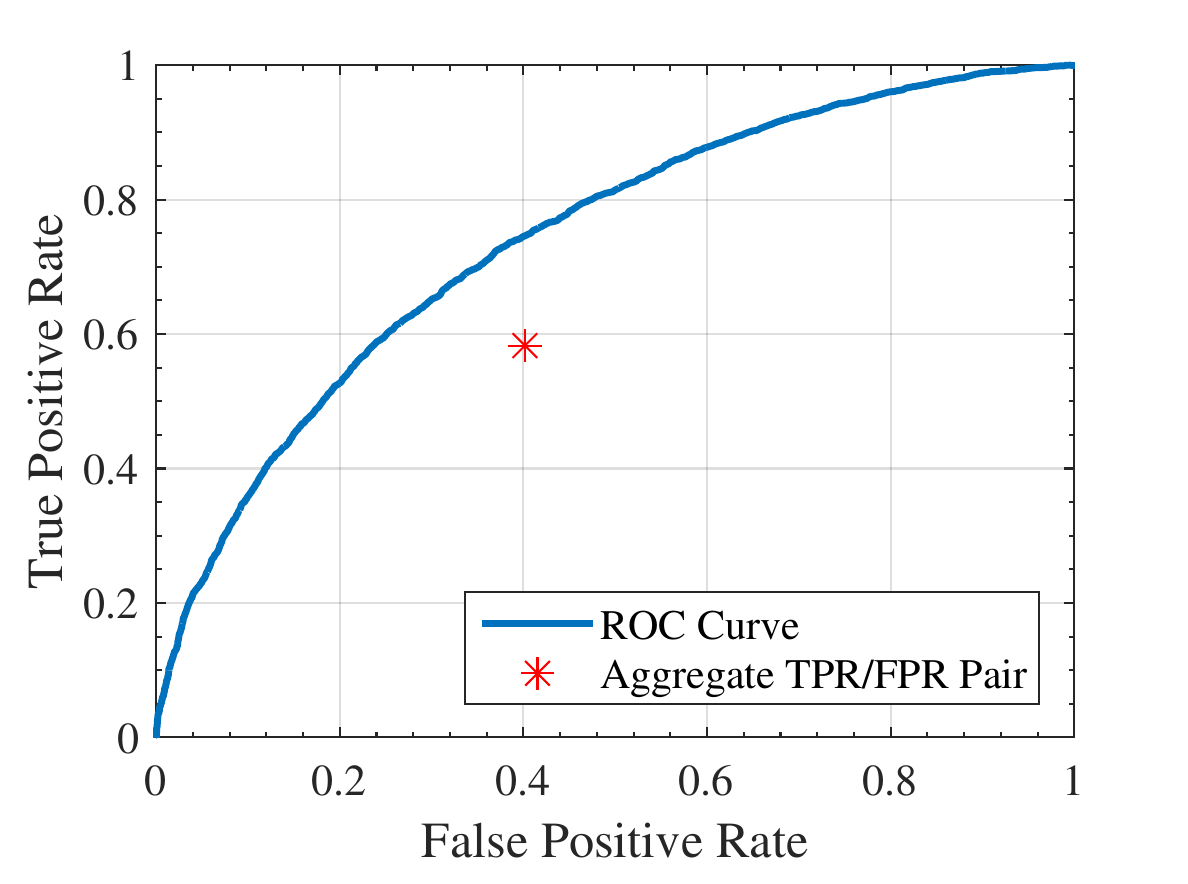}
  \end{center}
\end{figure}

Note that papers such as \cite{QJEbail} and \cite{elliott2013predicting} use
loss (cost) function to specify the tradeoff between under-detection and
over-detection rates, which is closely related to the ROC curve.

The focal point of our paper is to study the fundamental properties of ROC curves and proposes a statistical inference to the ROC curve. From \cite{bamber1975}, many papers such as \cite{fawcett1} and \cite{hand1}) study using the area under curve (AUC), a number instead of a curve, to measure the overall performance of a binary classifier. In this paper, we also present the inference of AUC and its implication for model selection.

Benchmarking the performance of machine algorithm, normally
presented by a ROC curve, with that of a human decision maker, presented by a
pair of false positive and true positive rates, is a common practice. For example, \cite{rajpurkar1} trained a 34-layer convolutional neural network to
process ECG sequences and compared its performance to 6 cardiologists. \cite{esteva1} proposed a deep convolutional neural network structure for skin cancer classification and claimed that the model outperforms the average dermatologist. \cite{kermany1} proposed an image-based deep learning model to
classify macular degeneration and concluded that it outperforms human.  The conclusions of these papers are mostly based on observing an empirical pair
of true positive and false positive rates that lie strictly below the ROC curve
formed by the machine classification algorithm, implying that machines can achieve a higher TPR for a given
FPR, or a lower FPR for a given TPR. However, an important message of our current paper is to \emph{caution against} such interpretations without a deeper
understanding of the human
decision making process: such findings can be rationalized not only by the
superior information quality of machine learning algorithms, but also by the \emph{incentive heterogeneity} of human decision makers who can be as
intelligent as machine learning algorithms in processing statistical information
from observational data.

To illustrate an issue of concern, consider
Figure \ref{figure
2}, in which a collection of human decision makers, denoted $j=1,\ldots,J$, all lie
approximately on the machine-learned ROC.
This is the case if they employed decision rules $\hat Y_i =
\mathds{1}\(\hat p\(X_i\) > c_i\)$ with the same $\hat p\(\cdot\)$ function but with
different individual cutoff points $c_i$.
Yet, after aggregating over all decision makers,
the aggregate TPR/FPR pair lies visibly below the ROC.
\begin{figure}
  \begin{center}
    \caption{Individual and Aggregate TPR/FPR Pairs: Perspective of Jensen's Inequality}\label{figure 2}
    \includegraphics[height=.32\textheight]{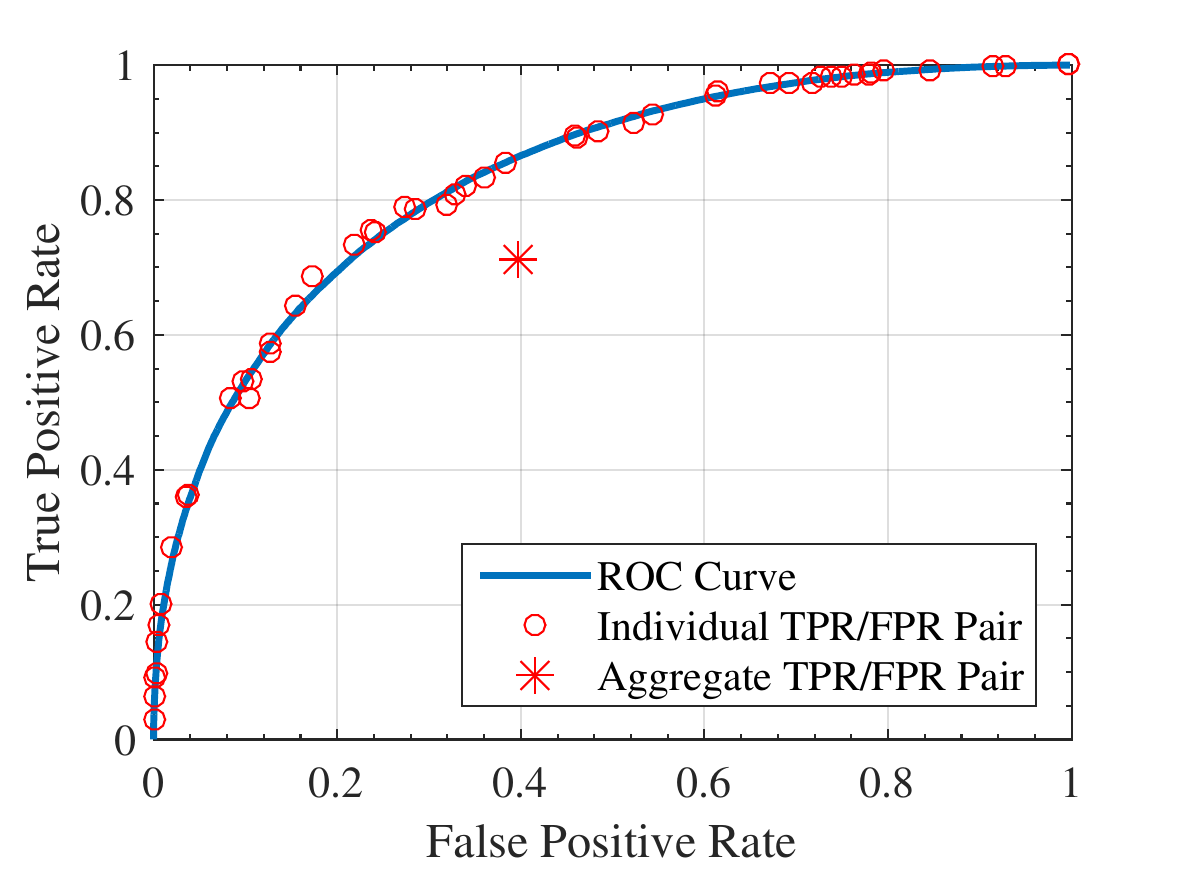}
  \end{center}
\end{figure}
This is an immediate artifact of Jensen's inequality due to the concavity of the
observed ROC, and bears no implication on the comparison between the qualities of the machine learning algorithm and human decision makers.
An optimal ROC is necessarily concave (Lemma \ref{concavity of roc lemma}).
This simple observation appears to have gone largely unnoticed by the
literature.

More precisely, as long as the collection of humans' individual TPR/FPR
points can be represented by a concave curve, the aggregated humans' TPR/FPR must
fall below the curve of humans' individual TPR/FPR
points.
Let $\alpha_j$ denote FPR and $\beta_j$ denote TPR,
and suppose they are related by $\beta = f\(\alpha\)$, where $f\(\cdot\)$ is
concave.  Then by Jensen's inequality:
 \begin{eqnarray*}
	 \bar \beta = \frac1J \sum_{j=1}^J \beta_j = \frac1J
	 \sum_{j=1}^n f\(\alpha_j\) < f\(\frac1J \sum_{j=1}^J
	 \alpha_i\) = f\(\bar \alpha\).
 \end{eqnarray*}

Furthermore, when incentive heterogeneity is present, a decision maker may also
set her cutoff value $c$ based on observed features $X$, denoted as
$c\(X_i\)$. As shown in subsection \ref{cutoff points}, in this case,
the TPR/FPR pair of a single decision maker is also below the optimal ROC curve.
Lemma \ref{ROCHetero} presents a formal proof. The foregoing discussion highlights the need to
understand and correctly interpret the statistical properties of ROC curves.

To motivate and empirically illustrate our theoretical discussion, we make use
of a data set of high-risk pregnancy diagnosis with more than a million
observations collected in the NFPC administered by the Chinese Ministry of Health. The checkup is offered free of
charge to newly-wed couples who are either expecting or who are planning to conceive.
The data set contains more than 300
features, including indicators from medical examinations and clinical tests,
individual and family history of diseases and drug usage, pregnancy history, and
demographic characteristics, etc. The label is whether the birth outcome is
normal or involves defects. Most importantly, the data set also contains the
diagnosis by doctors regarding the risk level of each pregnancy.

The rest of this paper are as follows. In section 2, we offer an in-depth analysis of the properties of the
ROC in the context of a statistical model and its relationship with loss
functions in decision making. Section 3 and 4 develop and study methods for the
statistical inference of ROC and for related model estimation and selection
issues. These two sections involve technical econometrics results that can be
skipped for readers mostly interested in economic insights.
Section 5 analyze the issues and caveats when comparing performance between humans and machine algorithms. In section 6, we present a detailed 
empirical application using the NFPC data set. Section 7 concludes.

\section{Neyman Pearson Lemma and Decision Rules}
In a standard statistical model, the data set
$Y_i,X_i,i=1,\ldots,n$ is typically assumed to be drawn i.i.d. (identically and
independently distributed) from an underlying population.
There is a true probability function
$p\(X_i\)= \mathbb{P}\(Y_i = 1 \vert X_i\)$, also known as the
propensity score function in the treatment effect literature. We will
invoke (uniform) law of large numbers and convergence in probability whenever
they may apply under plausible assumptions.
For this purpose, for $\mathbb X$ the support of $X_i$,
we assume that $\hat p\(X_i\)$ converges uniformly over $\mathbb X$ to a
deterministic limit function $q\(X_i\)$ when $n$
increases without bound:
\begin{eqnarray*}
\sup_{x \in \mathbb X} \left\vert \hat p\(x\) - q\(x\) \right\vert
	\overset{p}{\longrightarrow} 0, \quad\text{as}\quad n\rightarrow \infty.
\end{eqnarray*}
In the above, if the model that is used to estimate $\hat p\(X_i\)$ is correctly
specified, $q\(\cdot\) = p\(\cdot\)$. If a misspecified model is used to obtain
$\hat p\(X_i\)$, it is possible that $q\(\cdot\) \neq p\(\cdot\)$.

\subsection{Neyman Pearson Lemma}
Binary decision making is inherently related to parametric hypothesis
testing. In the terminology of hypothesis testing, $Y=0$ is often denoted as the null
hypothesis and $Y=1$ as the alternative hypothesis. The TPR/FPR pair corresponds
to the power (one minus Type II error) and the size (Type I error) of a test.
For a general classification rule $\hat Y_i = \mathds{1}\(X_i \in R\)$, where $R$ is
known as the rejection region in hypothesis testing, the law of large numbers
implies that
the sample TPR/FPR pair corresponding to each $R$
converges to their population analogs, denoted as PTPR and PFPR
{\begin{align}\begin{split}\nonumber
  \text{TPR} &\overset{\mathbb{P}}{\longrightarrow} \text{PTPR} \equiv \frac{ \mathbb{E}\[Y_i \mathds{1}\(X_i \in R\)\]}{p},\\
  \text{FPR} &\overset{\mathbb{P}}{\longrightarrow} \text{PFPR} \equiv \frac{ \mathbb{E}\[\(1-Y_i\) \mathds{1}\(X_i
\in R\)\]}{1-p}.
 \end{split}\end{align}}
In the above, $p = \mathbb{E}\[Y_i\] = \int p\(x\) f\(x\) \mathrm{d}x$ is the overall population portion of positive labels, 
where for simplicity we assume that $X$ has a density $f\(x\)$, which can be
broadly interpreted using generalized functions to include probability mass
functions for discrete $X$. 
By the law of iterated expectation,
\begin{align}
    \begin{split}\nonumber
    \text{PTPR} = \frac{1}{p}\int p\(x\) \mathds{1}\(x \in R\) f\(x\) \mathrm{d}x,\quad
    \text{PFPR} = \frac{1}{1 - p}
    \int \(1-p\(x\)\) \mathds{1}\(x \in R\) f\(x\) \mathrm{d}x
    .
    \end{split}
\end{align}
Recall the Bayes law
\begin{align}
    \begin{split}\nonumber
    	f\(X \vert Y=1\) = \frac{f\(X\) p\(X\)}{p},\quad
        f\(X \vert Y=0\) = \frac{f\(X\) \(1-p\(X\)\)}{1-p}.
    \end{split}
\end{align}
Therefore we can equivalently write
 {\begin{align}\begin{split}\nonumber
	\text{PTPR} = \int \mathds{1}\(X \in R\) f\(X \vert Y=1\) \mathrm{d}X,\quad
	\text{PFPR} = \int \mathds{1}\(X \in R\) f\(X \vert Y=0\) \mathrm{d}X.
 \end{split}\end{align}}
Consequently, PTPR is the probability of rejection under the
alternative hypothesis, namely the power of the test; PFPR is the probability of
rejection under the null hypothesis, namely the size of the test.  In the
population limit, the ROC is therefore a plot of power against size, for the collection of rejection areas
determined by $R = \mathds{1}\(q\(X_i\) > c\)$ when $c$ ranges over $\[0,1\]$.

The classical Neyman Pearson Lemma states that the collection of likelihood ratio tests
\begin{align}\begin{split}\nonumber
R_{NP}\(d\) = \biggl\{x:  	 \frac{f\(X \vert Y=1\)}{f\(X \vert Y=0\)}  >
	 d\biggr\},
\end{split}\end{align}
where $d \in \(0, \infty\)$ varies, are {\it most powerful tests} that
maximize power for whatever size it achieves. By the Bayes law,
write
\begin{align}
    \begin{split}\nonumber
    R_{NP}\(d\) = \biggl\{x:  	 \frac{p\(x\)}{1 - p\(x\)} >  d
    \frac{p}{1-p}\biggr\}
    = \biggl\{x:  	 p\(x\) >  c =  \frac{dp}{1-p + dp}\biggr\}.
     \end{split}
\end{align}
Consequently, the ROC corresponding to the decision rules
\begin{align}
    \begin{split}\nonumber
    	 \hat y = \mathds{1}\(p\(x\) > c\)\quad\text{for $c$ varying between $0$ and
    	 $1$,}
     \end{split}
\end{align}
has the Neyman-Pearson optimality that it lies weakly above the PTPR/PFPR pair
of any other decision rule $\hat y = \mathds{1}\(x \in R\)$ for any given $R$, or equivalently the
ROC of any alternative collection of decision rules.

Lemma \ref{concavity of roc lemma} in the Appendix shows that the optimal ROC is
necessarily a concave function. A non-optimal ROC, on the other hand, need not
be concave, and can even lie above the 45 degree line.

\subsection{Loss Function in Decision Making}\label{lossfunction}
The decision theoretic framework in \cite{elliott2013predicting} and \cite{QJEbail} balances the loss of false positive decision vs. the loss of false negative decision through a cost (loss) function. Since the choice decision depends only on
the utility difference and is invariant with respect to the addition of a normalizing constant,
without loss of generality we first consider
minimizing expected cost based on the following cost matrix, where the cost of correct classifications
is normalized to zero:
\begin{table}
  \caption{Loss Matrix} \label{loss matrix}
    \begin{center}
    \begin{tabular}{ l|cc }
      \toprule
      & $\hat Y=0, \text{Accept}$ & $\hat Y=1, \text{Reject}$ \\
      \midrule
      $Y=0$ & $0$ & $c_{0R}$ \\
      $Y=1$ & $c_{1A}$  & $0$\\
      \bottomrule
    \end{tabular}
  \end{center}
\end{table}
The (negative of the) expected loss governed by this cost matrix is then
\begin{align}
    \begin{split}\label{C0C1}
    &-\int \[C_{0R} \mathds{1}\(x\in R, Y=0\) + C_{1A} \mathds{1}\(x\in R^c,
    Y=1\)\] \mathrm{d}F\(x,Y\)
    \\
    =& -\int \biggl[\mathds{1}\(x\in R\) \(1-p\) C_{0R} f\(x \vert
     Y=0\) + \mathds{1}\(x\in R^c\) p C_{1A} f\(x \vert Y=1\)
     \biggr] \mathrm{d}x
    \\
    =& -C_{0R}  \mathbb{P}\(X\in R, Y=0\) + C_{1A} \mathbb{P}\(X \in R, Y=1\) + \text{const}
    \end{split}
\end{align}
where $\text{const}=-C_{1A} \mathbb{P}\(Y=1\)$ does not depend on the rejection region
$R$.  This takes the same form of a linear combination of PTPR and PFPR, i.e. $\phi \mathrm{PTPR} - \eta \mathrm{PFPR}$, where $\phi= p C_{1A}$ and $\eta=\(1-p\) C_{0R}$.
{\begin{align}
  \begin{split}
    \phi \text{PTPR} - \eta \text{PFPR}=& \frac{\phi}{p} \mathbb{P}\(Y=1, X \in R\) - \frac{\eta}{1-p}  \mathbb{P}\(Y=0, X \in R\) \\
    =& \int \[\mathds{1}\(X \in R\)
    \frac{\phi}{p}
    p\(X\)
    -
    \mathds{1}\(X \in R\)
    \frac{\eta}{1-p}
    \(1-p\(X\)\)\] f\(X\) \mathrm{d}X
  \end{split}
\end{align}}
By inspection, the optimal rule $R_c$ that maximizes this linear combination is given by
 {\begin{align}\begin{split}\label{c determined by eta and phi}
 R_c &= \biggl\{X:
\frac{\phi}{p}
 p\(X\) > \frac{\eta}{1-p}
\(1-p\(X\)\)\
\biggr\}
	 =\biggl\{X:
p\(X\) > c = \frac{\eta / (1-p)}{\phi/p + \eta / (1-p)}
\biggr\}.
 \end{split}\end{align}}
Therefore, for each $c$, the implied $R_c$ corresponds to the maximizer of
$\phi \text{PTPR} - \eta \text{PFPR}$ for some (nonunique) choice of
$\phi$ and $\eta$ which in turn determine $c$.
Consequently, for this $c$ and the resulting $R_c$, it is not possible to
choose an alternative $R$ in order to increase PTPR while keeping PFPR
unchanged, or to decrease PFPR while keeping PTPR unchanged.   In other words,
other classification rules that result in at least as much as this
PTPR will necessarily have the same or higher PFPR.
Equivalently, any other classification rule that results in at most
this PFPR will have the same or lower PTPR. Since the choice of $C_{0R}$ and $C_{1A}$ is subjective and lacks consensus, it is common to draw a curve of FPR and TPR according to different combination of $C_{0R}$ and $C_{1A}$ in an ROC curve.

\subsection{Additional Remarks}

\paragraph{Randomized tests} are sometimes used to obtain a desired size when the
observations $X$ have a discrete distribution, and has its analog in binary
classification as a randomized classification rule.

However, for size levels that
can be achieved by a function of the features,
randomization only serves to dilute power.
Specifically, a randomized classifier is a decision rule
 {\begin{align}\begin{split}\nonumber
	\hat Y = \phi\(X, U\) \in \{0,1\}\quad\text{were}\quad U \perp Y \vert X.
 \end{split}\end{align}}
The condition $U \perp Y \vert X$ rules out information about $Y$ that is not
already contained in $X$.
		
The resulting PFPR and PTPR can now be written as
\begin{align}\begin{split}\nonumber
\text{PTPR} =
	\frac{1}{p}\iint \phi\(X,U\) f\(U \vert X\) \mathrm{d}U p\(X\) f\(X\) \mathrm{d}X
\end{split}\end{align}
and
\begin{align}\begin{split}\nonumber
\text{PFPR} =
	\frac{1}{1 - p}\iint \phi\(X,U\) f\(U \vert X\) \mathrm{d}U \(1-p\(X\)\) f\(X\) \mathrm{d}X
 \end{split}\end{align}
Given $\phi$ and $\eta$, the maximizing $\phi\(X,U\)$ of a linear combination
of $\phi\text{PTPR} - \eta\text{PFPR}$ satisfies
\begin{align}
    \begin{split}\nonumber
    \int \phi\(X,U\) f\(U \vert X\) \mathrm{d}U  = \mathds{1}\(p\(X\) > c\),
    \end{split}
\end{align}
for $c$ determined by $\phi$ and $\eta$ as in \eqref{c determined by eta and
phi}. A solution is given by \eqref{c determined by eta and
phi}: $\phi\(X,U\) = \mathds{1}\(p\(X\) > c\)$.


\paragraph{Alternative representations} of the ROC are possible.
Consider the following question: conditional on a machine diagnosis of
being high risk, what is the probability
of actually being high risk. In other words, we would like to calculate, by
Bayes law:
\begin{align}\begin{split}\nonumber
	&\mathbb{P}\(Y = 1\vert \hat Y=1\) = \frac{\mathbb{P}\(Y=1, \hat Y=1\)}{\mathbb{P}\(\hat Y = 1\)}\\
  =& \frac{
  p \times \mathbb{P}\(\hat Y = 1\vert Y=1\)
  }{
  p \times \mathbb{P}\(\hat Y = 1\vert Y=1\) +
	(1-p) \times \mathbb{P}\(\hat Y = 1\vert Y=0\)
  } = \frac{
  p  \times \text{TPR}
  }{
	  p  \times \text{TPR} + (1-p) \times  \text{FPR}
  }
\end{split}\end{align}
In the above, $p$ is given by raw data. Next, each point on the ROC curve
corresponds to a $\text{TPR} / \text{FPR}$ pair. Therefore, for each point on the
ROC curve, we can calculate $\mathbb{P}\(Y = 1\vert \hat Y=1\)$. Connecting these numbers
will produce a ``posterior odds'' curve.

For example, with perfect classification, $\text{FPR}=0$ and $\text{TPR}=1$,
$\mathbb{P}\(Y = 1\vert \hat Y=1\) = 1$. With random guessing, the ROC curve is the 45
degree line, where $\text{FPR}=\text{TPR}$, then
$\mathbb{P}\(Y = 1\vert \hat Y=1\)=\mathbb{P}\(Y=1\)$, same as the raw sample unconditional
probability.

\paragraph{Sampling Errors} differentiates the ROC for the binary
classification problem from classical hypothesis testing problems.  In the size versus
power tradeoff in classical hypothesis testing, the conditional densities of the
features under both the null and the alternative, $f\(X \vert Y=0\)$ and
$f\(X \vert Y=1\)$, are known and fully specified. In contrast, in binary
classification these quantities, or equivalently the propensity score $p\(x\)$, need
to be estimated from the training data set and are thus subject to sampling errors.

In the population analysis we abstract away from such sampling errors, a topic that
we defer to in the estimation and inference sections. Precise knowledge of the correctly specified population
propensity score $p\(X\)$ is not feasible in finite samples. If the
features $X_i$ are supported on a small number of discrete points, $p\(x\)$ can
be estimated by the sample
frequency counts of $Y$ for each value in the support of $\mathbb X$:
 {\begin{align}\begin{split}\nonumber
	 \hat p\(X=k\) = \frac{\sum_{i=1}^n Y_i \mathds{1}\(X_i = k\)}{n_k},\quad n_k =
	 \sum_{i=1}^n \mathds{1}\(X_i = k\).
 \end{split}\end{align}}
When the features are continuously distributed or take a large number of
discrete values, regularization methods such as parametric assumptions, sampling
splitting, nonparametric regression, penalization, etc, are needed to reduce overfitting and to obtain $\hat p\(x\)$ that has out of
sample predictive powers.

The $\hat p\(x\)$ estimated from the training data is likely to be
misspecified in the holdout sample, or converge to a $q\(x\)$ that is
misspecified in the population.
In these situations, it may be difficult to find a
dominating ROC among several competing models, as their corresponding ROCs are
likely to cross each other.  Choosing among these ROCs involves a
subjective choice of criteria,
such as the F-score or the Area Under Curve (AUC). The sense in
which these alternative loss functions offer better model selection
criteria than the usual measurements (such as mean square error,
accuracy, or entropy divergence) remains to be investigated. \cite{huang1}
argued both theoretically and empirically for the advantage of using
AUC over accuracy.

\paragraph{The duality} between binary classification and classical hypothesis testing also suggests
a conditional version of the  Neyman-Pearson lemma that uses the
conditional density of $x_1$ given $x_2$: $f\(x_1 \vert x_2,
y=\{1,0\}\)$ and that employs a critical value
$c\(x_2\)$ that depends on $x_2$ only:
{\begin{align}\begin{split}\label{Rtilde}
\tilde R = \biggl\{x:  \frac{f\(x_1 \vert x_2,
		 y=1\)}{f\(x_1 \vert x_2, y=0\)}  > c\(x_2\)\biggr\}
\end{split}\end{align}
}
where $\tilde R$ is the collection of rejection areas.
This can be justified by minimizing expected conditional Bayesian loss
	given $x_2$, where the {\it prior} and the loss can depend on $x_2$
		 (only):
{\begin{align}\begin{split}\nonumber
	\tilde R =& \arg\min_R p_1\(x_2\)\ell_1\(x_2\)
	\int_{x_1 \in R} f\(x_1 \vert x_2, y=0\) \mathrm{d}x_1\\
	&\hspace{.1in}+
	\(1-p_1\(x_2\)\)\ell_0\(x_2\)
	\int_{x_1 \in R^c} f\(x_1 \vert x_2, y=1\) \mathrm{d}x_1.
\end{split}\end{align}
}
In the above, $\ell_0\(x_2\)$ is the loss of rejecting when the null is
true, and $\ell_1\(x_2\)$ is the loss of not rejecting when the null is
false. $p_1\(x_2\)$ is the {\it conditional} prior probability
of the null, and $\(1-p_1\(x_2\)\)$ is the {\it conditional} prior probability
of the alternative.\footnote{
Note that this is different from an unconditional
 Bayesian decision maker whose loss function depends only on
 $c\(x_2\)$,  and who makes use of the unconditional joint
 likelihood function of $x_1$ and $x_2$:
{
\begin{align}\begin{split}\nonumber
	\tilde R = \biggl\{x:  \frac{f\(x_1,  x_2 \vert
		 y=1\)}{f\(x_1, x_2 \vert y=0\)}  =
		\frac{f\(x_1 \vert x_2,
		 y=1\)}{f\(x_1 \vert x_2,
		 y=0\)} \frac{f\(x_2 \vert
		 y=1\)}{f\(x_2 \vert y=0\)}
		 > c\(x_2\)\biggl\}.
\end{split}\end{align}
}
}

While the rejection region now also takes a form of ({\ref{Rtilde}}), the determination of $c\(x_2\)$ also involves knowledge of $\frac{f\(x_2 \vert y=1\)}{f\(x_2 \vert y=0\)}$ in addition to the prior probabilities and loss function matrix.

\section{Statistical Inference of ROC Curves}
In this section we derive asymptotic pointwise confidence bands for an estimated
ROC to account for its sampling uncertainty.
%
Confidence bands can be useful for testing the statistical performance of a ROC
curve or testing whether one ROC lies above another.  These results are
technical in nature and are valid under conventional regularity conditions. For
brevity and clarify we present the main results. Detailed verification steps are
available from the authors upon request.
In reference of Figure \ref{figure 1}, confidence bands can be constructed vertically or horizontally. We begin with
the vertical construction, and then show that it is also valid in a horizontal
sense.

\subsection{Vertical Construction of Confidence Bands}

For tractability, we consider parametric models of $p\(X_i, \theta\)$ under
i.i.d sampling assumptions. A consistent estimate of $\hat\theta$ can be obtained by
maximum likelihood or other methods, such that
$\hat\theta\overset{p}{\longrightarrow} \theta_0$,  with an asymptotic linear
influence function representation:
\begin{align}\begin{split}\label{influence function for theta}
\sqrt{n}\(\hat\theta-\theta_0\) = \frac{1}{\sqrt{n}} \sum_{i=1}^n \kappa_i +
o_P\(1\),\quad\text{where}\quad  \kappa_i=\kappa\(y_i,x_i\)
\end{split}\end{align}
In the sample, the power and
size corresponding to the
classification rule based on $p\(X_i, \theta\)$ and a threshold value of $c$ are
given by
\begin{align}\begin{split}\nonumber
\hat\beta\(c\)=
\frac{
1/n \sum_{i=1}^n y_i \mathds{1}\(p\(x_i, \hat\theta\) > c\)
}{
\hat p
},\quad
\hat\alpha\(c\) =
\frac{
1/n \sum_{i=1}^n \(1-y_i\) \mathds{1}\(p\(x_i, \hat\theta\) > c\)
}{
1 - \hat p
}
\end{split}\end{align}
To simplify notation let $\hat c_\alpha = \hat\alpha^{-1}\(\alpha\)$ and
$\hat\beta_\alpha = \hat\beta\(\hat c_\alpha\)
= \hat\beta\(\hat\alpha^{-1}\(\alpha\)\)$. Non-continuity can be handled by
redefinining $\hat c_\alpha  = \inf\{
c: \hat\alpha\(c\) > \alpha \}$. The population analogs of the power and size
curves are defined by
\begin{align}\begin{split}\nonumber
\beta\(c\) = \frac{1}{p} \mathbb{E} \[ p\(X\) \mathds{1}\(p\(X, \theta_0\) > c\)\],
\quad
\alpha\(c\) = \frac{1}{1-p} \mathbb{E} \[\(1-p\(X\)\) \mathds{1}\(p\(X, \theta_0\) > c\)\].
\end{split}\end{align}
Similarly let $c_\alpha = \alpha^{-1}\(\alpha\)$ and $\beta_\alpha \equiv
\beta\(\alpha\) = \beta\(c_\alpha\)
= \beta\(\alpha^{-1}\(\alpha\)\)$.

The goal is to construct an asymptotic confidence inteval for $\beta_\alpha$
based on  $\hat \beta_\alpha$ for each $\alpha$, in the form of $\(\hat\beta_\alpha - \hat d,
\hat\beta_\alpha + \hat d\)$, to ensure a given approximate coverage probability
\begin{align}\begin{split}\label{horizontal confidence interval}
\lim\inf_{n\rightarrow\infty} \mathbb{P}\(
\hat\beta_\alpha - \hat d \leq \beta_\alpha \leq
\hat\beta_\alpha + \hat d \) \geq 1-\eta.
\end{split}\end{align}
Typical choices of $\eta$ are $1\%, 5\%$ and $10\%$. This is achieved by
deriving the asymptotic distribution of $\hat\beta_\alpha-\beta_\alpha$, which
in turn can be based on an influential function representation in the form of
\begin{align}\begin{split} \label{influence function for hat beta_alpha}
\sqrt{n} \(\hat\beta_\alpha - \beta_\alpha\) = \frac{1}{\sqrt{n}} \sum_{i=1}^n
\psi_i + o_{\mathbb{P}}\(1\),\quad\text{where}\quad
\psi_i=\psi\(y_i, x_i, \alpha\).
\end{split}\end{align}
It follows from \eqref{influence function for hat beta_alpha} that
\begin{align}\begin{split}\nonumber
\sqrt{n} \(\hat\beta_\alpha - \beta_\alpha\) \overset{d}{\longrightarrow}
N\(0,\sigma^2\),\quad\text{where}\quad \sigma^2= Var\(\psi_i\).
\end{split}\end{align}
For a consistent estimate $\hat\sigma^2 \overset{\mathbb{P}}{\rightarrow} \sigma^2$, we
can form $\hat d = \frac{1}{\sqrt{n}} \hat\sigma \Phi^{-1}\(1-\eta/2\)$.
In the following we present the procedure to derive \eqref{influence function for hat beta_alpha}
and to obtain $\hat\sigma^2$. Verifying \eqref{influence function for hat
beta_alpha} is also important for validating the use of resampling methods such
as the bootstrap for confidence interval construction.

To show \eqref{influence function for hat beta_alpha} we begin with
defining the sample moment conditions,
\begin{align}\begin{split}\nonumber
Q_n\(\theta,c,\beta\)
=& \frac1n \sum_{i=1}^n y_i \(\mathds{1}\(p\(x_i, \theta\) > c\) - \beta\),\\
P_n\(\theta,c,\alpha\)
=& \frac1n \sum_{i=1}^n \(1-y_i\) \(\mathds{1}\(p\(x_i, \theta\) > c\) - \alpha\).
\end{split}\end{align}
and their the population analogs
\begin{align}\begin{split}\nonumber
Q\(\theta,c,\beta\)
=& \mathbb{E}\[p\(X\) \(\mathds{1}\(p\(X, \theta\) > c\) - \beta\)\],\quad \\
P\(\theta,c,\alpha\)
=& \mathbb{E}\[\(1-p\(X\)\) \(\mathds{1}\(p\(X, \theta\) > c\) - \alpha\)\].
\end{split}\end{align}
For a correctly specified model, $p\(X\) = p\(X, \theta_0\)$.
By construction
\begin{align}\begin{split}\nonumber
Q_n\(\hat\theta,\hat c_\alpha,\hat\beta_\alpha\) = o_{\mathbb{P}}\(\frac1n\), \quad
P_n\(\hat\theta,\hat c_\alpha, \alpha\) = o_{\mathbb{P}}\(\frac1n\), \quad
Q\(\theta_0,c_\alpha,\beta_\alpha\) = 0, \quad
P\(\theta_0,c_\alpha, \alpha\) = 0.
\end{split}\end{align}
To account for the discontinuity of the sample moment conditions as a function
of the parameters, we assume that the parametric propensity score
$p\(X_i,\theta\)$ satisfies a typical
stochastic equicontinuity condition (Chapters 36 and 37 of \cite{newey_mcfadden}):
\begin{align}\begin{split}\nonumber
 P_n\(\hat\theta,\hat c_\alpha, \alpha\)
- P_n\(\theta_0, c_\alpha, \alpha\)
- P\(\hat\theta,\hat c_\alpha, \alpha\)
+ P\(\theta_0,c_\alpha, \alpha\) = o_{\mathbb{P}}\(\frac{1}{\sqrt{n}}\).
\end{split}\end{align}
The first term is $o_{\mathbb{P}}\(\frac1n\)$. Therefore
\begin{align}\begin{split}\nonumber
P_n\(\theta_0, c_\alpha, \alpha\)
+ P\(\hat\theta,\hat c_\alpha, \alpha\)
- P\(\theta_0,c_\alpha, \alpha\)
= o_{\mathbb{P}}\(\frac{1}{\sqrt{n}}\).
\end{split}\end{align}
Then by first order Taylor expansion,
\begin{align}\begin{split}\nonumber
\sqrt{n}\(\hat c_\alpha - c_\alpha\)
= -\(
\frac{\partial P\(\theta_0,c_\alpha,\alpha\)}{\partial c}
\)^{-1}
\[\sqrt{n}
P_n\(\theta_0,c_\alpha,\alpha\)
+
\frac{\partial P\(\theta_0,c_\alpha,\alpha\)}{\partial \theta}
\sqrt{n}\(\hat\theta-\theta_0\)
\] + o_{\mathbb{P}}\(1\),
\end{split}\end{align}
Under correct model specification, for $f_p\(\cdot\)$ the implied density of
$p\(X\)$ induced by $X$,
\begin{align}\begin{split}\nonumber
P_c\(\theta_0,c_\alpha,\alpha\)\equiv
\frac{\partial P\(\theta_0,c_\alpha,\alpha\)}{\partial c}
= - \(1 - c_\alpha\) f_p\(c_\alpha\).
\end{split}\end{align}
Without a specific function form for $p\(X,\theta\)$ we cannot provide a further analytic expression for
$P_\theta\(\theta_0,c_\alpha,\alpha\)=
\frac{\partial P\(\theta_0,c_\alpha,\alpha\)}{\partial \theta}$,
but it can be estimated consistently using finite sample numerical derivatives.

To derive \eqref{influence function for hat beta_alpha}
we continue to make use of stochastic equicontinuity
\begin{align}\begin{split}\nonumber
Q_n\(\hat\theta,\hat c_\alpha, \hat\beta_\alpha\)
- Q_n\(\theta_0, c_\alpha, \beta_\alpha\)
- Q\(\hat\theta,\hat c_\alpha, \hat\beta_\alpha\)
+ Q\(\theta_0,c_\alpha, \beta_\alpha\) = o_{\mathbb{P}}\(\frac{1}{\sqrt{n}}\).
\end{split}\end{align}
The first term is $o_{\mathbb P}\(\frac1n\)$, so that we can first order Taylor expand on
\begin{align}\begin{split}\nonumber
Q_n\(\theta_0, c_\alpha, \beta_\alpha\)
+ Q\(\hat\theta,\hat c_\alpha, \hat\beta_\alpha\)
- Q\(\theta_0,c_\alpha, \beta_\alpha\) = o_{\mathbb{P}}\(\frac{1}{\sqrt{n}}\),
\end{split}\end{align}
to conclude that,
\begin{align}\begin{split}\nonumber
\sqrt{n}\(\hat \beta_\alpha - \beta_\alpha\)
=& -\(
Q_\beta\(\theta_0,c_\alpha,\beta_\alpha\)
\)^{-1}
\biggl[\sqrt{n}
Q_n\(\theta_0,c_\alpha,\beta_\alpha\)
+
Q_\theta\(\theta_0,c_\alpha,\beta_\alpha\)
\sqrt{n}\(\hat\theta-\theta_0\)\\
&+
Q_c\(\theta_0,c_\alpha,\beta_\alpha\)
\sqrt{n}\(\hat c_\alpha-c_\alpha\)
\biggr] + o_{\mathbb{P}}\(1\).
\end{split}\end{align}
In the above
\begin{align}\begin{split}\nonumber
Q_\beta\(\theta_0,c_\alpha,\beta_\alpha\)=&\frac{\partial
Q\(\theta_0,c_\alpha,\beta_\alpha\)}{\partial \beta},\
Q_\beta\(\theta_0,c_\alpha,\beta_\alpha\)=\frac{\partial
Q\(\theta_0,c_\alpha,\beta_\alpha\)}{\partial \beta},\ \text{and}\\
Q_c\(\theta_0,c_\alpha,\alpha\)\equiv& \frac{\partial Q\(\theta_0,c_\alpha,\alpha\)}{\partial c}.
\end{split}\end{align}

When the model is correctly specified, such that $p\(x\) =
p\(x,\theta_0\)$, $\frac{\partial Q\(\theta_0,c_\alpha,\beta_\alpha\)}{\partial c}
= - c_\alpha f_p\(c_\alpha\)$. This can be combined with the representation for
$\sqrt{n}\(\hat c_\alpha - c_\alpha\)$ to write
\begin{align}\begin{split}\nonumber
\sqrt{n}\(\hat \beta_\alpha - \beta_\alpha\)
=& Q_\beta\(\theta_0,c_\alpha,\beta_\alpha\)^{-1}
\biggl[
\sqrt{n}
Q_n\(\theta_0,c_\alpha,\beta_\alpha\)
-\frac{c_\alpha}{1-c_\alpha}\sqrt{n}
P_n\(\theta_0,c_\alpha,\alpha\) \\
+&
\(Q_\theta\(\theta_0,c_\alpha,\beta_\alpha\)
-\frac{c_\alpha}{1-c_\alpha}
P_\theta\(\theta_0,c_\alpha,\alpha\)
\)
\sqrt{n}\(\hat\theta-\theta_0\)
\biggr]
+ o_{\mathbb{P}}\(1\).
\end{split}\end{align}
Noting that $\frac{\partial Q\(\theta_0,c_\alpha,\beta_\alpha\)}{\partial
\beta}=-p$, in combination with \eqref{influence function for theta}, we simplify to obtain
\eqref{influence function for hat beta_alpha}:
\begin{align}\begin{split}
\label{ROCkey}
\psi_i
=& \frac{1}{p}
\biggl[
y_i \(\mathds{1}\(p\(x_i, \theta\) > c\) - \beta\)
-\frac{c_\alpha}{1-c_\alpha}
\(1-y_i\) \(\mathds{1}\(p\(x_i, \theta\) > c\) - \alpha\)
\\
+&
\(Q_\theta\(\theta_0,c_\alpha,\beta_\alpha\)
-\frac{c_\alpha}{1-c_\alpha}
P_\theta\(\theta_0,c_\alpha,\alpha\)
\)
\kappa_i
\biggr] 
\end{split}\end{align}
We have therefore
obtained \eqref{influence function for hat beta_alpha}.
A consistent estimate $\hat\psi_i$ of the influence function $\psi_i$ can be obtained by
replacing unknown parameters and population quantities with sample analogs and
numerical derivatives, and subsequentially be used to form $\hat\sigma^2 = \frac1n \sum_{i=1}^n
\hat\psi_i^2$.

The pointwise convergence for each $\alpha$ can be strengthened to obtain uniform confidence bands over compact sets.
Under suitable regularity conditions, it can be shown that for  $a > 0$ and $b <
1$, $\hat\beta_\alpha - \beta_\alpha$ converges weakly and uniformly in  $\alpha \in
\[a,b\]$, implying that
\begin{align}\begin{split}\nonumber
\sup_{\alpha \in \[a,b\]} \sqrt{n} \lVert \hat\beta_\alpha - \beta_\alpha \rVert
\rightsquigarrow
\sup_{\alpha \in \[a,b\]} \lVert \mathbb{G}\(\alpha\)\rVert
\end{split}\end{align}
where $\mathbb{G}\(\cdot\)$ is a Gaussian process with covariance process
$Cov\(\psi\(y_i, x_i, \alpha\), \psi\(y_i, x_i, \alpha'\)\)$.
Convergence to the Gaussian limit also justifies the use of bootstrapping to form both pointwise and uniform confidence bands.
We defer a formal development of these results to a future study. Notions of
uniform weak convergence and the validity of bootstrap can be found in
\cite{kosorok2007introduction}.

The asymptotic linear
representation results hold regardless of whether $\hat\theta$ is
obtained using the same sample to compute the ROC, or is
estimated using a pre-training sample.
In the former case, $Q_n\(\theta_0,c_\alpha,\beta_\alpha\)$,
$P_n\(\theta_0,c_\alpha,\alpha\)$ are necessarily correlated with
$\sqrt{n}\(\hat \beta_\alpha - \beta_\alpha\)$.
In the later case, they are independent of each other under i.i.d
sampling. The results in this section also hold regardless of whether $p\(X_i, \theta\)$ is
correctly specified.

\subsection{Horizontal Construction of Confidence Bands}

The previous subsection constructs (pointwise) confidence band {\it vertically}
by defining $\hat\beta_l\(\alpha\) = \hat \beta_\alpha - \hat d$ and
$\hat\beta_u\(\alpha\) = \hat \beta_\alpha + \hat d$, such that for all
$\alpha \in \(0,1\)$, \eqref{horizontal confidence interval} holds. We now argue
that  \eqref{horizontal confidence interval} is also valid in a {\it horizontal}
sense. For this purpose, for each $\beta$ define $\alpha\(\beta\) = \beta^{-1}\(\beta\)$, and
define
$\hat\alpha_l\(\beta\)$
and $\hat\alpha_u\(\beta\)$ through the relation:
\begin{align}\begin{split}\nonumber
\hat\beta_u\(\hat\alpha_l\(\beta\)\) = \beta, \quad
\hat\beta_l\(\hat\alpha_u\(\beta\)\) = \beta.
\end{split}\end{align}
Then \eqref{horizontal confidence interval} is also horizontally valid, in the
sense that
\begin{align}\begin{split}
\label{pointwise CI vertical}
\forall \beta \in \(0,1\),\quad
\mathbb{P}\(\hat\alpha_l\(\beta\) < \alpha\(\beta\)
< \hat\alpha_u\(\beta\) \) \rightarrow 1 -\eta.
\end{split}\end{align}
For this purpose, it suffices to note from using monotonicity that
\begin{align}\begin{split}\nonumber
	\hat\alpha_l\(\beta\) < \alpha\(\beta\)
	\Longleftrightarrow \hat\beta_u\(\alpha\(\beta\)\) > \beta,\quad \hat\alpha_u\(\beta\) > \alpha\(\beta\)
	\Longleftrightarrow \hat\beta_l\(\alpha\(\beta\)\) < \beta.
\end{split}\end{align}
The following relations are therefore equivalent,
\begin{align}\begin{split}\nonumber
\hat\alpha_l\(\beta\) < \alpha\(\beta\)
	< \hat\alpha_u\(\beta\) \Longleftrightarrow
\hat\beta_l\(\alpha\(\beta\)\) < \beta < \hat\beta_u\(\alpha\(\beta\)\)
\Longleftrightarrow
\hat\beta_l\(\alpha\) < \beta\(\alpha\) < \hat\beta_u\(\alpha\).
\end{split}\end{align}
Consequently \eqref{horizontal confidence interval} and \eqref{pointwise CI
vertical} are equivalent statements.
Likewise, a horizontally constructed confidence interval is also vertically
valid.

\subsection{Simulation Results}\label{simulation1}

The asymptotic distribution derived above can be estimated analytically by
sample analogs or by resampling methods such as the bootstrap. We present a
small simulation exercise to illustrate their difference.

The data generating process is specified to be a logit model,
\begin{align}\begin{split}\nonumber
p\(X\) = \exp\(X'\beta\) / \(1 + \exp\(X'\beta\)\)
\end{split}\end{align}
where $X=\(X_1,X_2\)$, $\beta=\(1, -0.5\)$,
$X_{1}\sim N(2, 1)$, $X_{2} \sim N(0, 1)$, $B \sim Uniform(0, 1)$ and $Y = 1(p(X_{1}, X_{2}) > B)$. Note
that $X_{1}$ and $X_{2}$ are independent. In the simulation,
20,000 observations are randomly generated. We divide the data set into training set
and test set with a $1 : 1$ ratio. The training set is used to estimate the
parameters of the logit model, where the constant term is fixed to 0. Using the parameters
of the estimated model, we fit the test set to obtain its ROC curve. We then calculate the theoretical
values of the confidence bands of the ROC curve shown in (\ref{ROCkey})
analytically using sample analogs.

We also obtain the confidence bands for obtained ROC curve numerically by bootstrap. We divide each
bootstrap sample into two halves, one for training and one for test. After estimating the logit
model based on the training set, we draw ROC curves based on the test set and estimated model parameters.
We bootstrap the whole sample 1,000 times and draw the 95\% confidence bands based on the bootstrapped
ROC curves. Figure \ref{SimulatedROCBand} shows that the theoretical and bootstrapped values
of ROC confidence bands are closed matched to each other.
 \begin{figure}
\begin{center}
\caption{Theoretical and Bootstrapped Confidence Bands of the ROC Curve}\label{SimulatedROCBand}
    \includegraphics[height=.32\textheight]{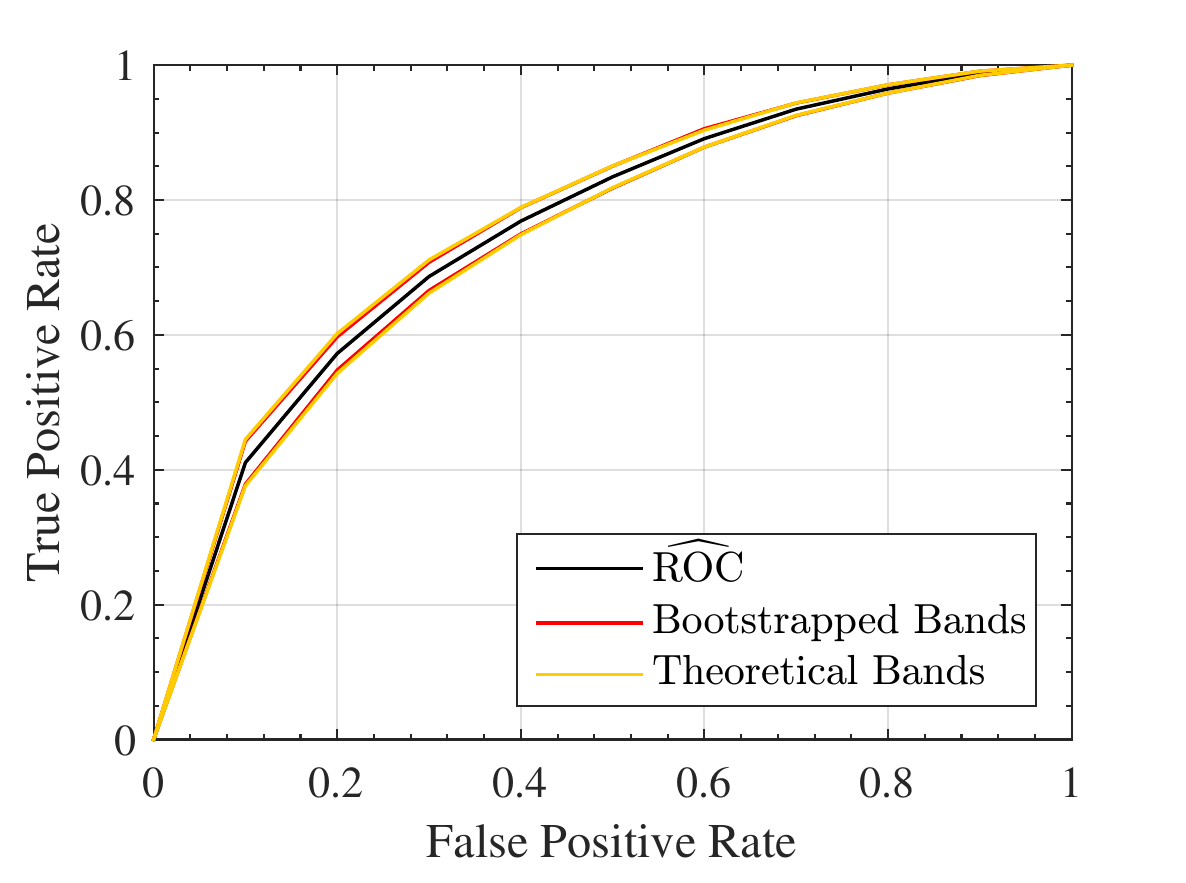}
\end{center}
\end{figure}
In future work we plan to conduct a full scale Monte Carlo exercise where the
above simulation is repeated numerous times and empirical coverage frequencies
are compared with nominal coverage probabilities.

\section{AUC and Model Comparison and Selection}

When the ROCs of two potentially misspecified propensity score models cross each other,
the area under the ROC curve (AUC) provides a heuristic crterion for model
selection in the machine learning literature. A formal statistical comparison
between the AUCs of estimated ROCs requires deriving their asymptotic
distributions. This section undertakes such a task for parametric propensity
score models.  First, we investigate how the AUCs can be used to obtain
parameter estimates.  Second, we derive the asymptoic distribution of sample
AUCs based on estimated parameters. Lastly, these results are used to provide
model selection tests and model selection criterion.

\subsection{Maximal AUC Estimator}\label{maximal auc estimator}

The parameters $\theta$ in $p\(X, \theta\)$ with parameters $\theta$ are
typically estimated by optimizing criterion (or loss) functions
such as the log likelihood (MLE or cross-entropy), or MSE (weighted or unweighted least square loss).
The AUC (Area under Curve) is an alternative criterion choice to cross-entropy
and MSE.  We show that the maximum AUC estimator is mathematically equivalent to the maximum rank
		 correlation estimator of \cite{sherman1993limiting}. While it
may be less efficient than MLE under correct specification,  it allows for
certain degrees of robustness, and is still consistent and asymptotically normal for a semiparametric
single index model.

The sample AUC corresponding to $\theta$ is given by
\begin{align}\begin{split}\nonumber
	\text{SAUC}\(\theta\) =
\int \hat \beta_\theta\(\alpha\) d\alpha =
\int \hat \beta_\theta\(\hat \alpha_\theta^{-1}\(\hat \alpha_\theta\)\) \mathrm{d}\hat\alpha_\theta =  \int
	\hat\beta_\theta\(c\) \frac{\mathrm{d}\hat\alpha_\theta\(c\)}{\mathrm{d}c} \mathrm{d}c
 \end{split}\end{align}
where
\begin{align}\begin{split}\nonumber
	\hat\beta_\theta\(c\) =& \frac{1}{\hat p} \frac1n \sum_{i=1}^n \mathds{1}\(p\(x_i, \theta\) >
	c\) y_i,
	\quad\text{and}\\
	\hat\alpha_\theta\(c\) =& \frac{1}{1-\hat p} \frac1n \sum_{j=1}^n \mathds{1}\(p\(x_j, \theta\) > c\)
	\(1-y_j\).
 \end{split}\end{align}
Using the notion of Dirac functions,
$\mathrm{d} \mathds{1}\(p\(x_j, \theta\) > c\) / \mathrm{d}c =
\delta_{p\(x_j,\theta\)}\(c\)$, we can write
\begin{align}\begin{split}\label{sample AUC}
    	\text{SAUC}\(\hat\theta\) = \frac{1}{n^2\hat p\(1-\hat p\)}
    	\sum_{i=1}^n \sum_{j=1}^n
    	\mathds{1}\(p\(x_i, \hat\theta\) > p\(x_j, \hat\theta\)\) y_i \(1-y_j\)
     \end{split}
 \end{align}
This takes the form of a U-process, which is $\sqrt{n}$ consistent and converges
asymptotically to a Gaussian process (\cite{sherman1993limiting} and
\cite{hoeffding1948class}).

The sample AUC converges to a population AUC, defined as
\begin{align}\begin{split}\nonumber
	\mathrm{PAUC}\(\theta\) =
\int\beta_\theta\(\alpha\) d\alpha=
\int \beta_\theta\(\alpha_\theta^{-1}\(\alpha_\theta\)\) \mathrm{d}\alpha_\theta =  \int
	\beta_\theta\(c\) \frac{\mathrm{d}\alpha_\theta\(c\)}{\mathrm{d}c} \mathrm{d}c
 \end{split}\end{align}
		 such that
\begin{align}\begin{split}\nonumber
	\beta_\theta\(c\) =& \frac{1}{p} \int \mathds{1}\(p\(x, \theta\) > c\) p\(x\) f\(x\) \mathrm{d}x
	\quad\text{and}\\
	\alpha_\theta\(c\) =& \frac{1}{1-p} \int \mathds{1}\(p\(x, \theta\) >
c\) \(1-p\(x\)\) f\(x\) \mathrm{d}x.
 \end{split}\end{align}
Using the notion of Dirac functions again,
we compute that
\begin{align}\begin{split}
\label{population parametric AUC}
	\text{PAUC}\(\theta\) = \frac{1}{p\(1-p\)}\iint \mathds{1}\(p\(x, \theta\) > p\(w,\theta\)\)
	p\(x\) \(1 - p\(w\)\) f\(x\)  f\(w\) \mathrm{d}x \mathrm{d}w.
 \end{split}\end{align}
This integral would be maximized with respect to the indicator
		 function if the indicator is turned on whenever
$p\(x\) \(1 - p\(w\)\) > p\(w\) \(1 - p\(x\)\),$
	equivalently or whenever $p\(x\) > p\(w\)$. Under correct specification,
		 this can obviously be achieved when $\theta=\theta_0$, where
$p\(x,\theta_0\)=p\(x\) > p\(w\) = p\(w,\theta_0\).$
Therefore, by standard M-estimator arguments (e.g. \cite{newey_mcfadden}) the maximum AUC estimator
is consistent under correct specification and suitable sample regularity conditions.

\subsection{Single Index Model}

A common semiparametric specification of $p\(x_i, \theta\)$ is a
single index model, where
$p\(x_i, \theta\) = \Lambda\(x_i'\beta\(\theta\)\)$ 
and $\Lambda\(\cdot\)$ is strictly increasing but may be unknown. In this case,
\begin{align}\begin{split}\nonumber
	\text{SAUC}\(\theta\) = \frac{1}{n^2\hat p\(1-\hat p\)}
	\sum_{i=1}^n \sum_{j=1}^n
	\mathds{1}\(x_i'\beta\(\theta\) > x_j'\beta\(\theta\)\) \mathds{1}\(y_i > y_j\)
 \end{split}\end{align}
since $\mathds{1}\(y_i > y_j\) = y_i \(1-y_j\)$ for binary $y_i, y_j$. As in
\cite{sherman1993limiting} $\beta\(\theta\) = \(1,\theta\)$ in order to reduce
the dimension of the parameter space by $1$ since $\beta$ is only identified up
to a multiplicative factor.

This is exactly the maximum rank correlation estimator of \cite{sherman1993limiting}.
The theoretical consistency of $\hat\theta\overset{\mathbb{P}}{\rightarrow} \theta_0$ only requires that $\theta_0$ is the only point in $\theta\in\Theta$ such that
\begin{align}
    \begin{split}\nonumber
    	x'\beta\(\theta\) > w'\beta\(\theta\)\quad\text{if and only if}\quad p\(x\) > p\(w\),
    	\quad \forall x, w.
    \end{split}
\end{align}
\cite{sherman1993limiting} develops the asymptotic properties of $\hat\theta$
and shows that $\sqrt{n}\(\hat\theta-\theta_0\)$ is asymptotically linear
under the assumption of a correctly
specified model of the form
\begin{align}\begin{split}\nonumber
y_i = D\(g\(x_i'\beta\(\theta\), \epsilon_i\)\)
 \end{split}\end{align}
where $g\(\cdot,\cdot\)$ is smooth and strictly increasing in both arguments,
$D\(\cdot\)$ is weakly increasing, and $\epsilon_i \perp x_i$. However, the high
level framework in \cite{sherman1993limiting}  is readily generalizable to allow
for potential misspecification where the single index assumption does not hold.

Under misspecification we assume that sufficient regularity conditions hold such that
a pseudo true value is uniquely defined in the population:
\begin{align}\begin{split}\nonumber
\theta^* = \arg\max_{\theta \in \Theta} \mathbb{E} \[\mathds{1}\(X_i'\beta\(\theta\) >
X_j'\beta\(\theta\)\) \mathds{1}\(Y_i > Y_j\)\].
 \end{split}\end{align}
We note the sections 1-5 of \cite{sherman1993limiting} hold generally without
the single index assumption.
Only section 6 of \cite{sherman1993limiting}, the exact form of the asymptotic variance,
the influence function, and the Hessian matrix, need to be generalized to allow
for misspecification. The appendix provides these general forms,
which specialize to Section 6 of \cite{sherman1993limiting} under the single index model,
and verifies that
\begin{align}\begin{split}\label{single index influence function}
\sqrt{n}\(\hat\theta-\theta^*\) = \frac{1}{\sqrt{n}} \sum_{i=1}^n \kappa_i +
o_P\(\frac{1}{\sqrt{n}}\),
 \end{split}\end{align}
where $\hat\theta = \arg\max_{\theta\in\Theta} \text{SAUC}\(\theta\)$.  Note
that \eqref{single index influence function} generally holds for $\hat\theta$
obtained from optimizing many criterion functions such as the KLIC, with
a corresponding influence function $\kappa_i$ that is dependent on the
criterion function.

The single index model is semiparametric and does not require knowledge of the
transformation function $\Lambda\(\cdot\)$. If $\Lambda\(\cdot\)$ is known (so
is $p\(x, \theta\)$) and the model is fully parametric, the AUC can
alternatively be estimated by the sample analog of \eqref{population parametric
AUC}:
\begin{align}\begin{split}\label{sample parametric AUC}
\text{SAUC}_p\(\hat\theta\)
= \frac{1}{\hat p\(1-\hat p\)} \frac{1}{n^2} \sum_{i=1}^n \sum_{j=1}^n
\mathds{1}\(p\(x_i,
\hat\theta\) > p\(x_j,\hat \theta\)\) p\(x_i, \hat\theta\) \(1 - p\(x_j,
\hat\theta\)\).
\end{split}\end{align}
When $p\(x_j,\hat \theta\)$ is correctly specified and when $\hat\theta$ is
obtained by MLE, \eqref{sample parametric AUC}
is more efficient than \eqref{sample AUC} for estimating \eqref{population parametric AUC}.
The asymptotic efficiency gain can be analytically calculated.

\subsection{Asymptotic Normality of AUC}

This section makes further use of the U-process convergence
results in \cite{sherman1993limiting} to derive the influence function
representation and asymptotic normality of the sample estimated AUC. These
results are used in the next section to provide the basis for AUC-based model comparison
tests and model selection criteria. Previous works by
\cite{hanley1982meaning} and \cite{hsieh1996nonparametric} that treated the SAUC
as two-sample U-statistics (see \cite{lehmann2004elements}) do not account for
the parameter estimation uncertainty and the random denominator containing
  $\hat{p} = 1/n\sum^{n}_{i = 1}y_{i}$ as derived in Section \ref{maximal auc estimator}.

Let $z_{i} = \(x_{i}, y_{i}\)$ and $\omega\(z_{i}, z_{j}, \theta\) =
\mathds{1}\(p\(x_{i}, \theta\) > p\(x_{j}, \theta\)\)\mathds{1}\(y_{i} >
y_{j}\)$. Recall that
\begin{align}\begin{split}\label{sample parameter estimated AUC}
	\hat A
=& \text{SAUC}\(\hat\theta\) = \frac{1}{n^{2}\hat{p}\(1 -
\hat{p}\)}\sum^{n}_{i = 1}\sum^{n}_{j = 1}\omega\(z_{i}, z_{j}, \hat\theta\),\\
A 
=& \text{PAUC}\(\theta^*\)  = \frac{1}{p\(1-p\)} E  \omega\(z_{i}, z_{j}, \theta^*\).
\end{split}\end{align}
The goal is to obtain the linear representation and limiting distribution of
$\hat A$:
\begin{align}\begin{split}\label{asymptotic normality of SAUC}
    \sqrt{n}\(\hat{A} - A\) =
\frac{1}{\sqrt{n}} \sum_{i=1}^n \xi_i + o_{\mathbb{P}}\(\frac{1}{\sqrt{n}}\),\quad
    \sqrt{n}\(\hat{A}- A\)  \overset{d}{\longrightarrow} N\(0, Var\(\xi_i\)\).
\end{split}\end{align}
To obtain \eqref{asymptotic normality of SAUC} we begin by writing
\begin{align}\begin{split}\nonumber
    \sqrt{n}\(\hat{A} - A\) = \frac{\sqrt{n}}{n^{2}\hat{p}\(1 -
\hat{p}\)}\sum^{n}_{i = 1}\sum^{n}_{j = 1}\(\omega\(z_{i}, z_{j}, \hat\theta\) - A\hat{p}\(1 - \hat{p}\)\).
\end{split}\end{align}
Note that
$\hat{p}\(1 - \hat{p}\) \overset{\mathbb{P}}{\longrightarrow} p\(1 - p\),$
it suffices to show that
for some $\zeta_i$
\begin{align}\begin{split}\nonumber
    \frac{\sqrt{n}}{n^{2}}\sum^{n}_{i = 1}\sum^{n}_{j = 1}\(\omega\(z_{i},
z_{j}, \hat\theta\) - A\hat{p}\(1 - \hat{p}\)\)
= \frac{1}{\sqrt{n}} \sum_{i=1}^n \zeta_i + o_{\mathbb{P}}\(\frac{1}{\sqrt{n}}\).
\end{split}\end{align}
By arguments of Slutsky's lemma, then
\begin{align}\begin{split}\nonumber
\sqrt{n}\(\hat{A} - A\)=
\frac{1}{\sqrt{n}} \sum_{i=1}^n \frac{1}{p\(1-p\)} \zeta_i + o_{\mathbb{P}}\(\frac{1}{\sqrt{n}}\).
\overset{d}{\longrightarrow} N\(0, \frac{Var\(\zeta_i\)}{p^{2}\(1 - p\)^{2}}\).
\end{split}\end{align}
Noting that
\begin{align}\begin{split}\nonumber
    \hat{p}\(1 - \hat{p}\) & = \frac{1}{n}\sum^{n}_{i = 1}y_{i} - \frac{1}{n^{2}}\sum^{n}_{i = 1}\sum^{n}_{j = 1}y_{i}y_{j}  = \frac{1}{n^{2}}\sum^{n}_{i = 1}\sum^{n}_{j = 1}\[y_{i}\(1 - y_{j}\)\],
\end{split}\end{align}
and that $\hat{p}\(1 - \hat{p}\) = \frac{1}{n^{2}}\sum^{n}_{i = 1}\sum^{n}_{j =
1}\(\hat{p}\(1 - \hat{p}\)\)$,
we can rewrite as
\begin{align}\begin{split}\nonumber
    \frac{\sqrt{n}}{n^{2}}\sum^{n}_{i = 1}\sum^{n}_{j = 1}\(\omega\(z_{i},
z_{j}, \hat\theta\) - A\hat{p}\(1 - \hat{p}\)\) =
\frac{\sqrt{n}}{n^{2}}\sum^{n}_{i = 1}\sum^{n}_{j =
1}\underbrace{\[\omega\(z_{i}, z_{j}, \hat\theta\) - Ay_{i}\(1 - y_{j}\)\]}_{\eta\(z_{i}, z_{j}, \theta\)},
\end{split}\end{align}
where, using the definition of $\omega\(z_{i}, z_{j}, \theta\)$ and note that $\mathds{1}\(y_{i} > y_{j}\) = y_{i}\(1 - y_{j}\)$,
\begin{align}\begin{split}\nonumber
    \eta\(z_{i}, z_{j}, \theta\) = \(\mathds{1}\(p(x_{i}, \theta) > p\(x_{j}, \theta\)\) - A\)y_{i}\(1 - y_{j}\)
\end{split}\end{align}
If we define,
\begin{align}\begin{split}\nonumber
\hat Q\(\theta\) = \frac{1}{n^{2}}\sum^{n}_{i = 1}\sum^{n}_{j = 1}
\eta\(z_i, z_j, \theta\)\quad\text{and}\quad
Q\(\theta\) = \mathbb{E}\[\eta\(z_i, z_j, \theta\)\] \equiv 0,
\end{split}\end{align}
then we can invoke the U-process stochastic equicontinuity results in
\cite{sherman1993limiting}:
\begin{align}\begin{split}\nonumber
\hat Q\(\hat\theta\) - \hat Q\(\theta^*\) - Q\(\hat\theta\) + Q\(\theta^*\) =
o_{\mathbb{P}}\(\frac{1}{\sqrt{n}}\),
\end{split}\end{align}
such that, using \eqref{single index influence function}, where $\hat\theta$ is
obtaining from optimizing a general criterion function that may or may not be
the SAUC,
\begin{align}\begin{split}\nonumber
\hat Q\(\hat\theta\) - Q\(\theta^*\)
=& \hat Q\(\theta^*\) - Q\(\theta^*\) + Q\(\hat\theta\) - Q\(\theta^*\)  +
o_{\mathbb{P}}\(\frac{1}{\sqrt{n}}\)\\
=&\frac{1}{n^{2}}\sum^{n}_{i = 1}\sum^{n}_{j = 1}
\eta\(z_i, j_j, \theta^*\)
+ \frac{\partial}{\partial\theta} Q\(\theta^*\) \frac1n \sum_{i=1}^n \kappa_i +
o_{\mathbb{P}}\(\frac{1}{\sqrt{n}}\)\\
=&\frac1n \sum_{i=1}^n \(\eta_{1}\(z_{i}, \theta^*\) + \eta_{2}\(z_{i}, \theta^*\)
+ \frac{\partial}{\partial\theta} Q\(\theta^*\)  \kappa_i\) +
o_{\mathbb{P}}\(\frac{1}{\sqrt{n}}\).
\end{split}\end{align}
where by H\'ajek projection,
\begin{align}\begin{split}\nonumber
    \eta_{1}\(z_{i}, \theta\)  = \mathbb{E}_{z_{j}}\[\eta\(z_{i}, z_{j},
\theta\)\], \quad \eta_{2}\(z_{j}, \theta\)  = \mathbb{E}_{z_{i}}\[\eta\(z_{i}, z_{j}, \theta\)\],
\end{split}\end{align}
We therefore conclude that \eqref{asymptotic normality of SAUC} holds with
\begin{align}\begin{split}\nonumber
\zeta_i = \eta_{1}\(z_{i}, \theta^*\) + \eta_{2}\(z_{i}, \theta^*\)
+ \frac{\partial}{\partial\theta} Q\(\theta^*\)  \kappa_i\quad\text{and}\quad
\xi_i = \frac{1}{p\(1-p\)} \zeta_i.
\end{split}\end{align}
If $\hat\theta$ is the maximum AUC estimator, then
$\frac{\partial}{\partial\theta} Q\(\theta^*\) =0$, so that the last term
vanishes from $\zeta_i$, in which case parameter estimation uncertainty vanishes
asymptotically in first order.  On the other hand, parameter uncertainty is
first order important when $\hat\theta$ optimizes another criterion function
that differs from the AUC.

\eqref{asymptotic normality of SAUC} can be used for
analytically constructing asymptotic tests by estimating sample analogs of $\xi_i$,
or for justifying the validity of resampling procedures.

The results in the previous three sections assume that the entire sample is used
for both parameter and AUC estimation, or that a fixed sample splitting scheme
is employed. A random split of the sample into parameter estimation and AUC
calculation can be accommodated by introducing random indicators $D_i$, such
that $D_i=1$ denotes the parameter estimation subsample and $D_i=0$ denotes the
AUC calculation subsample. Previous results can be readily generalized by
replacing \eqref{single index influence function} with $D_i \kappa_i$,  the
double summation in \eqref{sample parameter estimated AUC} with
\begin{align}\begin{split}\nonumber
	\sum^{n}_{i = 1}\sum^{n}_{j = 1} \(1-D_i\) \(1-D_j\) \omega\(z_{i}, z_{j}, \hat\theta\),
\end{split}\end{align}
and corresponding changes in related influence functions. In general, sample
splitting increases the variance of the estimated AUC.

\subsection{Model Comparison Tests and Model Selection Criteria}

Following a vast literature in econometrics and statistics, results derived in the
last two sections provide the basis for constructing model selection tests
and forming model selection criteria. This section provides a brief overview.

When at least one model is correctly specified, any criterion function
(such as cross-entropy or KLIC) combined with suitable penalization can be
used to select the most parsimonious correct model, or to be used to form
specification test statistics for individual models.
Under the assumption of potentially misspecification of all models under
consideration, the end result may depend on the choice of criterion
functions even asymptotically.
KLIC or MSE are popular choices in econometrics and statistics, while in
computer science, AUC has been advocated for this purpose. We therefore focus on
the AUC.

If the AUC is used as the criterion functions both for estimating the parameters of
the competing models and for testing and selecting between models, then the
generaly results from \cite{vuong1989likelihood}, and \cite{hong2003generalized},
apply.  It is also possible that a different criterion function, such as cross
entropy, is used to estimate parameters before the use of the AUC criterion
function to compare or select between competing models.
The sample can be split between estimation and model testing and
selection, or the same sample can be employed for both purposes.
Both cases can be handled similarly using asymptotic linear representations.


Consider two competing models with parameters $\theta$ and $\vartheta$, and
corresponding sample AUCs $\hat A_1\(\hat\theta\)$ and   $\hat A_2\(\hat\vartheta\)$, then  it follows from
\eqref{asymptotic normality of SAUC}
that
\begin{align}\begin{split}\label{first order difference in AUC}
\hat A_1\(\hat\theta\) - \hat A_2\(\hat\vartheta\) = \(A_1\(\theta^*\) -
A_2\(\vartheta^*\)\) + \frac{1}{n} \sum_{i=1}^n \(\xi_i^1 - \xi_i^2\) +
o_{\mathbb{P}}\(\frac{1}{\sqrt{n}}\).
 \end{split}\end{align}
A test of the null hypothesis of $A_1\(\theta^*\) = A_2\(\vartheta^*\)$
between two nonnested models can rely asymptotically on the nondegenerate distribution of
$\xi_i^1 - \xi_i^2$. For nested models where $\xi_i^1-\xi_i^2$ may be degenerate, a second
order expansion of the parameter estimation uncertainty may be
required to obtain a nondegenerate limit distribution.

A model selection criterion typically takes the form of $\hat A\(\hat\theta\) + \kappa_n
\dim\(\theta\)$ which penalizes the number of parameters in order to be
parsimonious. Consistent model selection, in the sense of choosing the most
parsimonious and best fit model with probability model, is achieved by requiring
the penalization term $\kappa_n$ to satisfy $\kappa_n \rightarrow
0$ and  $\sqrt{n} \kappa_n \rightarrow \infty$. The latter requirement can be
relaxed to $n \kappa_n \rightarrow \infty$ between nested models.

\subsection{Simulation Results}
We continue to use the data set of Section \ref{simulation1} as an illustrating
example. First the true model coefficients (i.e. 1 and -0.5) are assumed to be
known and are used to estimate the theoretical asymptotic standard deviation of
SAUC based on analytic sample analogs as 0.00632.
We also bootstrap test data set for 1,000 times, calculating AUCs for each
iteration at the true coefficients. The bootstrapped standard deviation of the
SAUC thus obtained is 0.00629, which accords with the analytic estimate of the
theoretical value in the absence of parameter estimation uncertainty and random
sample splitting.

Next we incorporate parameter estimation uncertainty.
Table \ref{maxAUC} lists the estimated coefficients of the logit model
(with a fixed intercept as 0) using a maximal likelihood approach and
a maximum AUC (MAUC) approach using the training data set.
We estimate the standard deviation of the two coefficients in the logit model
using 1000 bootstrap repetitions of the data sample.
The coefficient estimates are both close to the true model parameters, but slightly different.
 For the first coefficient, the point estimate of MAUC is closer to the true
parameter than MLE is, but with a larger standard deviation estimate. The
reverse is true for the second coefficient. Finally we report the AUC estimates
using the mean of the bootstrap
based on MLE and MAUC, and their bootstrap standard deviations. The AUC
estimate based on MAUC is slightly larger than that based on MLE. The bootstrap
standard deviations are essentially identical,  and show that a substantial
component (more than 50\%) of the standard deviation is due to the uncertainty
from parameter estimation and from random sample splitting.

\begin{table}
\caption{Parameter Estimations using MLE and Maximal AUC Approaches} \label{maxAUC}
\begin{center}
\begin{tabular}{l|ccc}
  \toprule
 & True Model& MLE& Maximal AUC\\
   \midrule
  X1 (mean)& 1& 1.0045& 0.9992\\
  X1 (std)& NA& 0.0178& 0.0381\\
  X2 (mean)& -0.5& -0.5502& -0.5688\\
  X2 (std)& NA& 0.0301& 0.0184\\
  AUC (mean)& 0.7683& 0.7689& 0.7694\\
  AUC (std)& NA& 0.0151 & 0.0151\\
    \bottomrule
\end{tabular}
\end{center}
\end{table}

Table \ref{ModelSelection} reports an AUC-based model selection exercise between
two misspecfied models given the data generating process in section \ref{simulation1}.
The first (M1) is a logit model with $p(X_{1}) =\frac{\exp(\theta_1 X_{1})}{1 +
\exp(\theta_1 X_{1})}$; the model (M2) is a logit model with $p(X_{2})
=\frac{\exp(\theta_2 X_{2})}{1 + \exp(\theta_2 X_{2})}$. We bootstrap the whole
data set 1,000 times. In each iteration, we randomly partition the whole data
set in two halves, one as the training set for parameter estimation and the
other one as the test set for AUC calculation. We then estimate M1 and M2
separately and obtain their AUCs (A1 and A2) accordingly. Finally, we calculate
the mean AUC spreads and corresponding standard deviation, using both
the theoretical first order approximation in \eqref{first order difference in AUC} and using
bootstrapping.
We obtain a significant $z$ score: $z=\frac{\hat A_1\(\hat\theta\) - \hat
A_2\(\hat\vartheta\)}{std(\hat A_1\(\hat\theta\) - \hat A_2\(\hat\vartheta\))}$,
which rejects the null hypothesis that M1 is equivalent to M2 in favor of the
directional alternative hypothesis that M1 is a better model.
\begin{table}
\caption{Model Selection} \label{ModelSelection}
\begin{center}
\begin{tabular}{l|cc}
  \toprule
  & Bootstrap& Theoretical\\
  \midrule
  A1 (mean)&0.7341& 0.7314\\
  A2 (mean)&0.6214& 0.6191\\
  A1-A2 (mean)&0.1127& 0.1124\\
  A1-A2 (std)&0.0102 & 0.0103\\
  \bottomrule
\end{tabular}
\end{center}
\end{table}

As in section \ref{simulation1}, a full scale Monte Carlo simulation to gauge
the bias, variance and mean square errors of competing estimators is left for
future investigation.

\section{Human Decisions and Machine Decisions}\label{performance}

In a dataset, we typically observe a sample of the triples $\(Y_i, D_i, X_i\),
i=1,\ldots,n$, where $Y_i$ is the label outcome variable such as whether the
birth is normal or defect, $D_i$ is the predicted
value for $Y_i$ by a human decision maker such as a doctor, and $X_i$ is the set of
features that can be used to predict $Y_i$ or $D_i$. A machine learning
algorithm learns a model $\hat p\(X_i\)$ using the training subset of $Y_i,
X_i$,  and forms a ROC curve based on the prediction rules $\hat Y_i = 1\(\hat
p\(X_i\) > c\)$ when $c$ varies between $0$ and $1$.

An aggregate human TPR/FPR pair is calculated using $Y_i,
D_i, i=1,\ldots,n$ as
\begin{align}\begin{split}\nonumber
\text{TPR} = \frac{\sum_{i=1}^n Y_i D_i}{\sum_{i=1}^n Y_i}\quad\text{and}\quad
\text{FPR} = \frac{\sum_{i=1}^n \(1-Y_i\) D_i}{\sum_{i=1}^n \(1-Y_i\)}
\end{split}\end{align}

A popular approach to benchmark the performance of human (doctor) decisions is
to see whether the (TPR, FPR) pair lies above or below the AI ROC curve. That
the human (TPR, FPR) pair lies below machine ROC is interpreted as evidence
that AI dominates human decision making, and vice versa. In the former case,
there exists a point on the machine ROC with the same FPR as human decision
makers but higher TPR than  human decision makers. Alternatively, there exists
another point on the machine ROC with the same TPR as human decision makers but
with lower FPR. No point on the machine ROC is dominated by the human
(TPR, FPR) pair in the sense of having lower values in both ratios. See, for
example, \cite{QJEbail}.

These interpretations are based on the strong assumption that doctors employ the
same model and use the same cutoff value $c$ to form predictions. That the
machine generated ROC dominates human  (TPR, FPR) pair is then attributed to
doctors using less information or a misspecified model.  In reality,
both information and incentive heterogeneity are likely to be present among
doctors. The information set that humans
base their diagnosis on may differ from
the feature set used in the machine learning algorithm. For example,
doctors may have observed patients' complexion and discovered other information
when conversing with patients, which are not recorded as features in the data set.
Furthermore, even if doctors employ the same model using the same information
set, they might differ in the cutoff value $c$ because of the difference in
their perception regarding the tradeoff between type I and type II errors.
This section explores how heterogeneous incentive and information asymmetry complicate
the comparison between machines and humans.

Our data set has a panel structure where multiple observations are available for
each doctor, such that $n$  (TPR, FPR) pairs can be drawn for each of the $n$
doctors. The panel structure does not eliminate concerns about information and
incentive heterogeneity, and is subject to the criticism based on Jensen's
inequality outlined in section \ref{introduction}. However, it does provide a
useful source of identification that we will explore in this section.

Furthermore, using $D_i, X_i, i=1,\ldots,n$, we can train a machine learning
model for predicting human decisions $D_i$ based on the features $X_i$, and
generate a probabilistic estimate of $q\(X_i\) = \mathbb{P}\(D_i = 1 \vert X_i\)$. In
other words, $q\(X_i\)$ is a model of ``predictded doctors''. Next,
we can use $q\(X_i\)$ to forecast $Y_i$ based on the rules
$\hat Y_i = 1\(q\(X_i\) > c\)$ where $c$ varies between $0$ and $1$, and form a
resulting ROC curve. This will be called the ``predicted doctor ROC''.
The rest of this section explores the relation between the aggregate (TPR, FPR)
pair, the machine learned ROC (MROC), and the predicted doctor ROC (PROC) under
different incentive and information structures. We maintain the assumption that
the machine can accurately learn the correct $p\(X_i\) = \mathbb{P}\(Y=1 \vert X_i\)$.

\subsection{The PROC with homogeneous information and incentive}

First we note that regardless of the information and incentive structure,
MROC and PROC might coincide even when $p\(X\)$ and $q\(X\)$
differ.  In particular, whenever $q\(x\)$ is a strictly increasing
transformation of $p\(x\)$, their corresponding
ROCs must coincidence, because by monotonicity, there exists a function
$\bar c\(t\)$ such that $q\(x\) > t$ if and only if $p\(x\) > \bar
c\(t\)$. Therefore a point on the MPOC corresponding to threshold
$t$ maps into a point on the PROC corresponding to threshold $\bar
c\(t\)$. Lemma \ref{same ROC if and only if} establishes the converse. For
example, $q\(x\)$ is a monotonic transformation of $p\(x\)$ when they are both
monotonic functions of a scalar $x$.

Next consider a situation where humans use the same model based on the observed
$X$, and are incentive homogeneous by using the same cutoff value $c_0$. In
other words, for some $\bar p\(x\)$, $D_i = 1\(\bar p\(X_i\) \geq c_0\)$. Then
the humans' ROC curve shrink to a TPR/FPR singleton, which also
coincides with the aggregate humans' TPR/FPR pair. $\bar p\(x\)$ can be
misspecified and need not be related to the true $p\(x\)$.
For example, Figure \ref{figure 7} is based on a model in which
$x\sim Uniform\(0,1\)$, $p\(x\) = x$, but humans mistakenly think
that $\bar p\(x\) = \vert \sin\(10x\)\vert$. The homogeneous humans'
decision rule is $\hat Y = \mathds{1}\(\bar p\(x\) > .5\)$.
Lemma \ref{ShrinkSingleton}  suggests that the PROC
reduces to a single point that is also the humans' aggregate TPR/FPR pair,
both of which lie strictly below the MROC. It will be on the MROC when $\bar
p\(\cdot\) = p\(\cdot\)$, or when $\bar p\(\cdot\)$ is an increasing
transformation of $p\(\cdot\)$.
Empirically, it is very unlikely that neither information nor incentive
heterogeneity is present.
\begin{figure}
\begin{center}
\caption{(Misspecified) Homogeneous Information}\label{figure 7}
\includegraphics[height=.32\textheight]{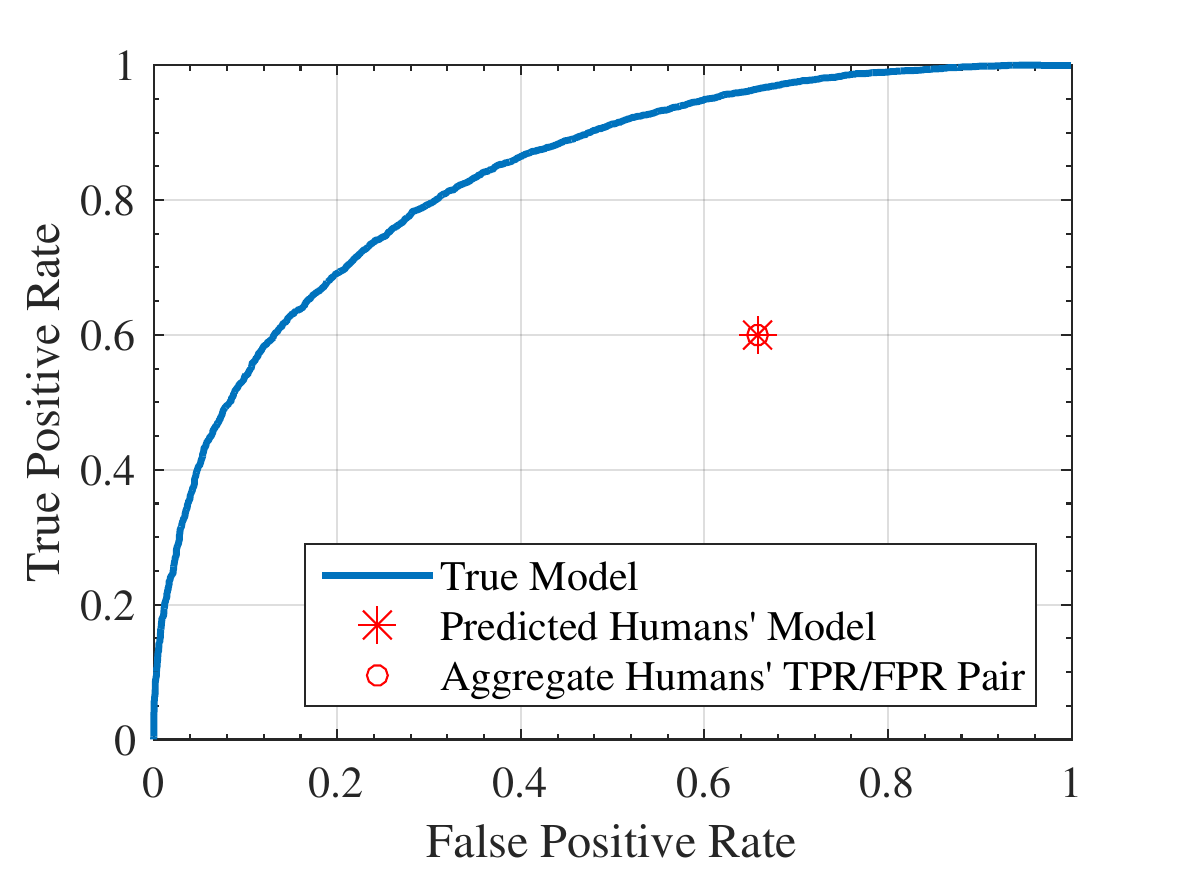}
\end{center}
\end{figure}

\subsection{Incentive Heterogeneity} \label{aggregatePair}

This section considers incentive heterogeneity (only, without information
asymmetry). Specifically, suppose all humans have access to the same information of $X$,
and they know the true $p\(X\)$. However,
each human employs their own idiosyncratic fixed cost for the optimal tradeoff
between true positives and
false positives, implying that humans are heterogeneous in the
cutoff values that they use in making predictions.

Moreover, each individual decision maker may not have a constant preference, and may
change her cutoff threshold $c$ based on observed features. In this case, each decision
maker does not need to be represented by a single point on the ROC curve.
For example, \cite{currie3} find that patients' demand on C-section
delivery can influence doctor's decision. If certain groups of patients with
similar characteristics have similar C-section demands, the cutoff value $c$ will
depend on these characteristics of patients $x$. \cite{chandra2011} find that hospitals
treat similar patients differently due to consideration of commercial benefits.
Therefore $c$ depends on both hospitals and the location of patients.
In the context of our empirical application in section \ref{application},
suppose $x$ is a one dimensional variable denoting age.  While the probability of birth
defect is likely to be strictly increasing in $x$, doctors may place a high
weight on diagnosing normal older couples correctly: they would not want older but normal
couples to forgo any chance of conception.
Then $c\(x\)$ can also be strictly increasing in $x$.
As a result, doctors may tend to diagnose younger couples as
abnormal but older couples as normal. The next subsection
illustrates this in more details.

If this type of incentive heterogeneity exists across the sample of decision
makers, or within each individual, the aggregate
TPR/FPR pair can lie below the optimal ROC curve, or even below the 45 degree
line, even if humans have better information processing capacities than
machines.  Lemma \ref{ROCHetero} provides a formal proof.

\subsubsection{Cutoff Thresholds}\label{cutoff points}
If the loss functions $C_{0R}$ and $C_{1A}$ for a decision maker in (\ref{C0C1})
depends on observed features $x$ as proposed by \cite{elliott2013predicting},
(\ref{c determined by eta and phi}) also shows that $c$ is no more a constant, but
instead depends on $x$.

Moreover, the predicted humans' ROC curve may also differ from
the true (and optimal) MROC curve, {\it even if} humans know the true $p\(X\)$.
For example, let $X\sim Uniform\(0,1\)$, $p\(X\) = X$, $\eta \sim Uniform\(0,1\)$, $Y=\mathds{1}\(p\(X\) > c\(X,\eta\)\)$,
\begin{align}\begin{split}\nonumber
c(X,\eta)= \eta \[\(X < .5\)+ \(X > .75\)\] + \eta X^2 \(.5 < X <
.75\).
\end{split}\end{align}
So that the ``predicted human'' becomes
{\begin{align}\begin{split}\nonumber
\bar p\(x\) = x  \[\(x < .5\)+ \(x > .75\)\] + 1/x \(.5 < x < .75\).
\end{split}\end{align}
}
Figure \ref{figure 5} plots the PROC curve of this example, part of which lies
strictly below the MROC.
\begin{figure}
\begin{center}
\caption{Incentive Heterogeneity, Incentive Distributions (Cutoff Threshold) Dependent on Features}\label{figure 5}
\includegraphics[height=.32\textheight]{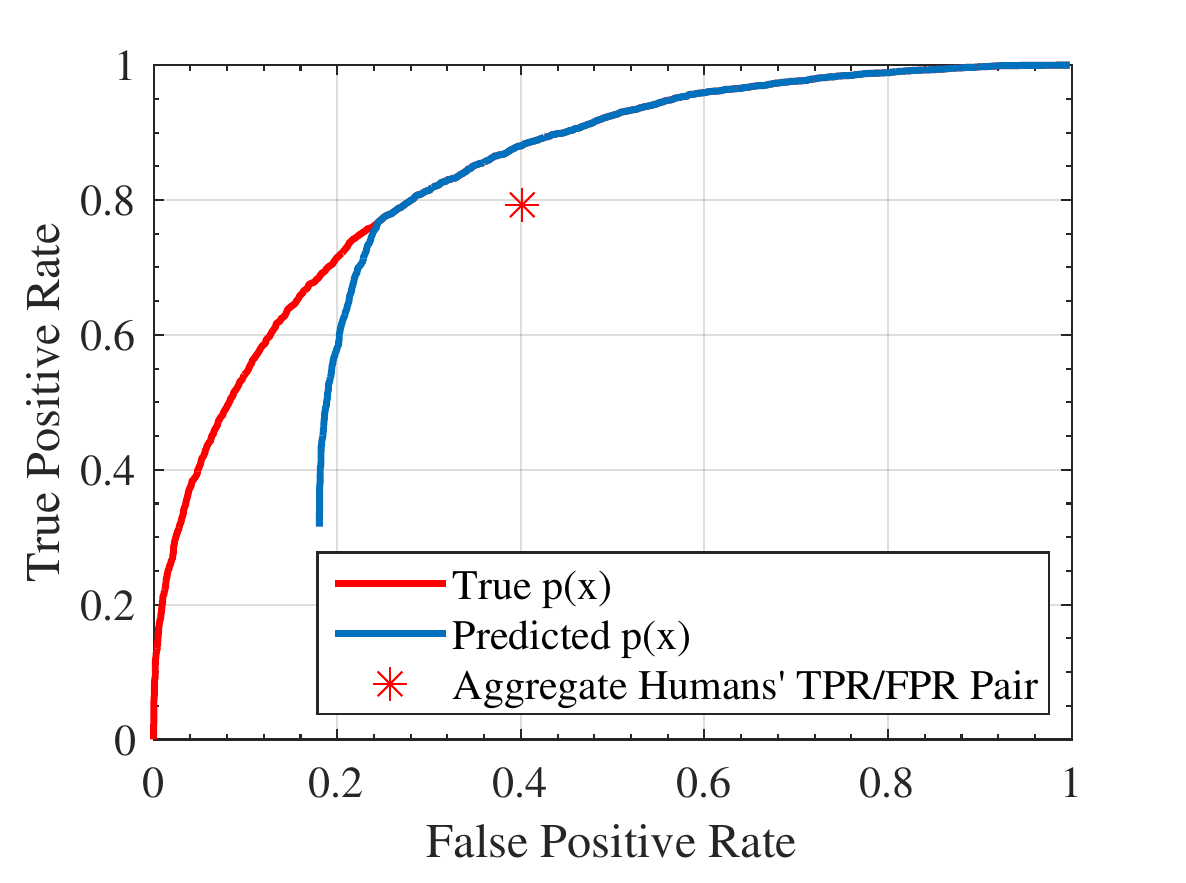}
\end{center}
\end{figure}

Incentive heterogeneity may be controlled by conditioning on a subset of the
features. If we postulate  that the features $x$ can be partitioned into $x_1, x_2$
where the cutoff function $c\(x_2\)$  only
depends on $x_2$, then we can conduct the analysis conditional
on $x_2$, and construct both ROC curves and the humans' pairs across
$x_1$ for each given $x_2$
to examine their relations. This requires plotting
multiple conditional ROCs across $x_1$, for each level of
$x_2$, to see if humans' pairs of TPR/FPR conditional on $x_2$
fall on each of the corresponding (to each $x_2$) ROC curves. If $x_2$ takes a small number of finite values, the analysis can
be done for each value of $x_2$. If $x_2$ is continuous, then some type of local smoothing will be
needed to generate these graphs. Unless $p\(x_1, x_2\)$ does not depend on $x_2$, if we
had plotted the aggregate ROC corresponding to $p\(x_1, x_2\)$
jointly in $x_1, x_2$, it will lie above the individual ROCs corresponding to
each level of $x_2$. As such, each of humans' PTPR/PFPR
pairs might still lie underneath the aggregate ROC due to incentive heterogeneity in
decision making.

\subsubsection{Uncover Incentive Heterogeneity from Data}\label{uncover}
A natural follow up question is the extent to which individual heterogeneity in
decision making can be recovered from empirical data. To be more specific, assume that humans have the same and correct information
as the machine learning algorithm does, i.e. have knowledge of the correct propensity
score function $p\(X\)$. In this subsection, we temporarily define the cutoff value as $S$ (instead of $c$), which is a random variable whose distribution may depend on $X$, with a
conditional cumulative distribution function denoted as $F_{S\vert X}\(S \vert
x\)$. First of all, from the empirical data, a learning algorithm recovers an
estimate of $p\(X\)$ using the real outcome as the label.  Then using
the decision of humans as the label, a learning algorithm can recover
{\begin{align}\begin{split}\nonumber
	q\(x\)= \mathbb{P}\(S <  p\(x\) \vert x\) = F_{S \vert X}\(p\(x\) \vert x\).
\end{split}\end{align}
}
Given a data set of $X_i$, $p_i=p\(X_i\)$, and $q_i=q\(X_i\)$, this translates into a relation of
{\begin{align}\begin{split}\nonumber
	q_i = h\(p_i, X_i\) \quad\text{where}\quad h\(p,x\) = F_{S\vert X}\(p,
	x\).
\end{split}\end{align}
}
A supervised learning algorithm can be used to estimate $h\(p,x\)$ using $q_i$
as the label and the tuple of $\(p_i, X_i\)$ as the features, and can be expected to
perform well if $h\(p,x\)$ satisfies suitable regularity or continuity
conditions.
There are two applications of the learned $h\(p,x\)$ function. First, when $S$ is independent of $X$, where $q\(X\) =
F_S\(p\(X\)\)$, $q_i$ is then a monotonically increasing transformation of
$p_i$. Therefore, by testing for monotonicity, we can infer whether the
distribution of the random threshold depends on observed features or not. Second, if
a panel data set is available, by comparing $h\(p,x\)$ for different human decision makers, we can
explicitly recover the incentive heterogeneity of their decision making.

\subsection{Information Asymmetry}

Next we abstract away from incentive heterogeneity and focus on information
asymmetry.  Suppose that in addition to the observed features $X$,
humans also make decision based on $U$, where $U$ is
not observed by the researchers or the machine learning algorithm. For example, $U$ could be
the facial complexion of the patient that the doctor observed during the day of
the clinic visits.  Let $p\(X,U\) = \mathbb{P}\(Y = 1 \vert U, X\)$.

Consider first a case where humans know the correctly specified propensity score
$p\(X,U\)$, and use a homogeneous decision rule $\hat Y = \mathds{1}\(p\(X,U\) > c\)$. In
this case, because $p\(X,U\)$ correctly uses more information than $p\(X\)$
does, the (unobserved) ROC formed by
$p\(X,U\)$ lies strictly above the MROC formed by the machine's propensity score
{\begin{align}\begin{split}\nonumber
p\(X\) = \int p\(X,U\) f\(U \vert X\) \mathrm{d}U.
\end{split}\end{align}}
Furthermore, the
aggregate human TPR/FPR pair must lie on the ROC formed by
$p\(X,U\)$, and hence above the machine's MROC curve. Lemma \ref{aboveROC}
provides the proof. Therefore, if humans are incentive \emph{homogenous} in decision making,
and their aggregate TPR/FPR pair is below the machine's ROC curve, one can conclude that
humans either do not know the true propensity score or use wrong information in
making decisions.

Next, consider a situation in which the human has potentially more, but
\emph{wrong} information, and employs a misspecified model
$\bar p\(x,u\)$.  In this case,
the predicted human's ROC curve is
defined by the ROC corresponding to the ``pseudo'' propensity score function:
\begin{align}\begin{split}\nonumber
\bar p\(x\) =
\mathbb{E}_{U \vert X} \[\mathds{1}\(\bar p\(X,U\)>c\)\] =
\int \mathds{1}\(\bar p\(X,U\)>c\)  f\(U \vert X\) \mathrm{d}U
\end{split}\end{align}
which is likely to differ from $p\(X\)$. The implied PROC might coincide with
the MROC, or might below it.

For example, let the true model be
$p\(x,u\) = \exp(x) / \(1+\exp(x)\),$
but human decision maker uses the model, for $u$ uniformly distributed on
$\(0,1\)$,
{\begin{align}\begin{split}\nonumber
\bar p\(x,u\) = \frac{\exp(x+u)}{1+\exp(x+u)}.
\end{split}\end{align}
}
In this model, despite using the wrong information model, the predicted human's
propensity score $\bar p\(x\) = \mathbb{E}_u \[\mathds{1}\(\bar p\(x,u\) \geq c\)\]$ is a monotonic
function of $x$, and
therefore has a PROC curve based on $x$ that coincides with
that of $p\(x\)$ (see Appendix \ref{same ROC if and only if}).

To generate divergence between the predicted humans' PROC curve and the MROC curve from
the true model, we would need to introduce either reverse monotonicity or
nonmonotonicy between $\bar p\(x\)$ and $p\(x\)$, or multidimensional
features $x$. If the human's information model has the wrong sign in $x$:
{\begin{align}\begin{split}\nonumber
\bar p\(x,u\) = \frac{\exp(-x+u)}{1+\exp(-x+u)}.
\end{split}\end{align}
}
Then the predicted humans' PROC curve and the true MROC curve differ in Figure \ref{figure 9}.
\begin{figure}
\begin{center}
\caption{Misspecified Heterogeneous and Non-monotonic Information}\label{figure 9}
\includegraphics[height=.32\textheight]{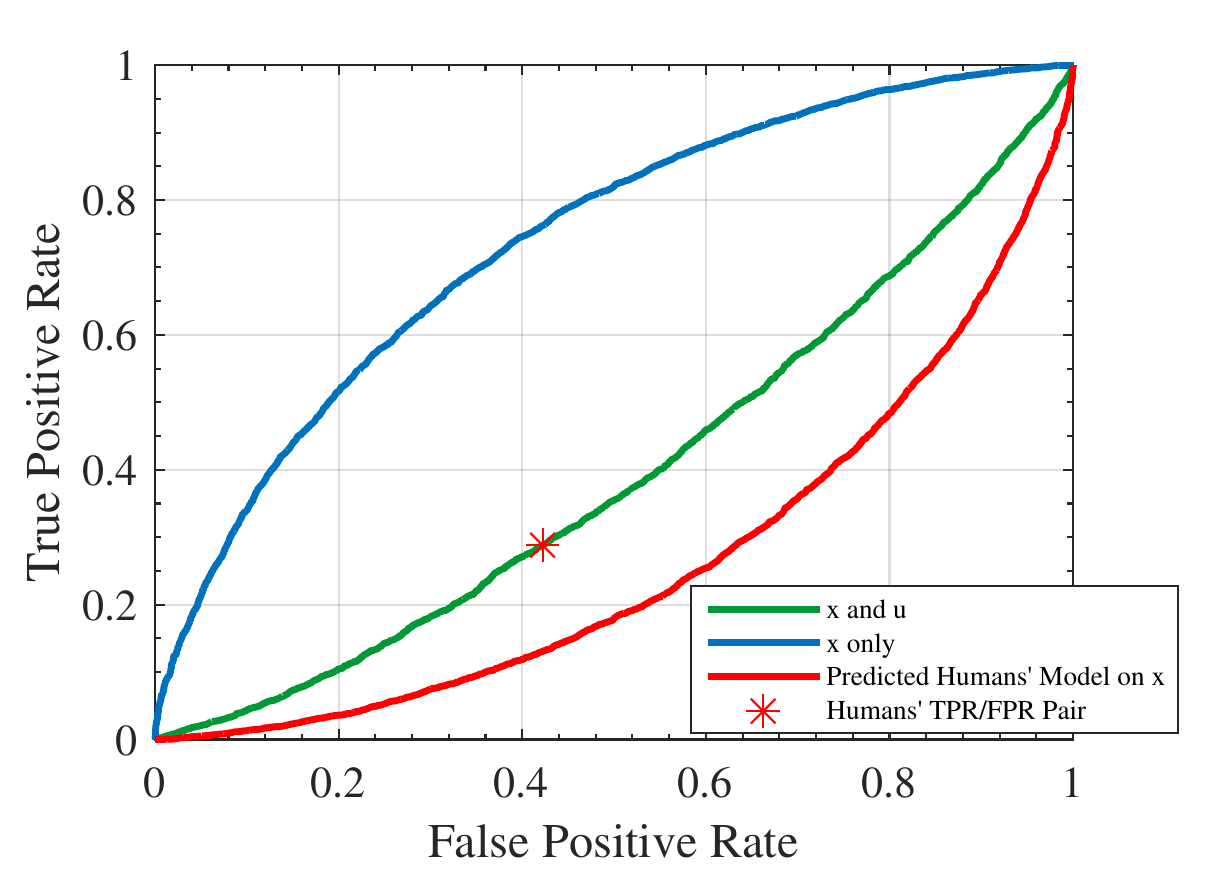}
\end{center}
\end{figure}
Next consider a two dimensional feature case. Let the true model be
{\begin{align}\begin{split}\nonumber
p\(x_1, x_2, u\) = \exp(x_1+x_2)/(1+\exp(x_1+x_2)),\quad x_1 \sim N\(0,1\),\quad x_2 \sim N\(0, 0.09\).
\end{split}\end{align}
}
And let the human's (wrong) information model be
{\begin{align}\begin{split}\nonumber
\tilde p\(x_1,x_2, u\) = \exp(x_1-x_2+u)/(1+\exp(x_1-x_2+u)),\ \text{where } u \sim N\(0,4\),
\end{split}\end{align}
}
Then we observe ROC curves in Figure \ref{figure 10}.

\begin{figure}
\begin{center}
\caption{Misspecified Heterogeneous and Multi-dimensional Information}\label{figure 10}
\includegraphics[height=.32\textheight]{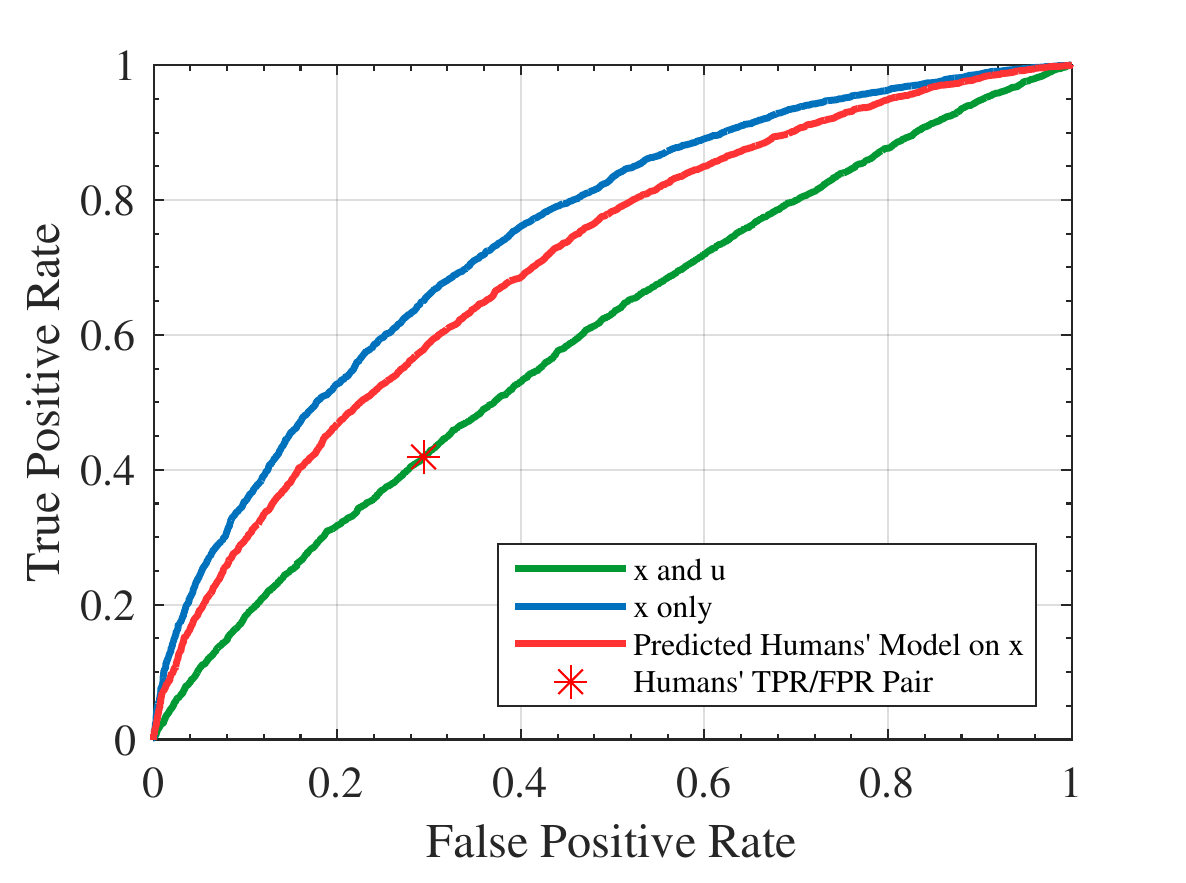}
\end{center}
\end{figure}

\section{Application: Birth Defect Decisions}\label{application}
To motivate and empirically illustrate our theoretical discussion, we make use
of a survey data set of high risk pregnancy diagnosis with more than a million
observations as an application of our findings.

\subsection{Data}
Our data is derived from the National Free Pre-Pregnancy Checkups (NFPC) in
Henan Province.
The NFPC is a population-based health survey of reproductive-aged couples and was administered
in China from January 1, 2014 to December 31, 2015 to expecting couples or
couples who plan to have children. The data set include most of the information
available to the doctor at the time of the diagnosis, including age and other
demographic characteristic, results from medical examination and clinical test,
individual and family history of diseases and drug use, pregnancy history, as
well as separate lifestyle and other living environment information for both
spouses. Altogether 282 features are available for each observation in the
sample. The label of the data indicates whether the pregnancy outcome
is normal (y = 0) or involves birth defects (y = 1). Of these samples,
4239(2.2\%) couples have an adverse pregnancy outcome.

In addition, the data set also includes doctor's pregnancy risk assessments.
The doctors' diagnosis is coded as 4 levels-0 for normal, 1 for high-risk in terms of female,
2 for high high-risk in terms of male, and 3 for high risk in terms of both female and male.
Unless otherwise noted, we regard 1, 2, 3 of doctor's assessment as a diagnosis
of risky outcome, and 0 as a diagnosis of normal pregnancy outcome. A total of
20.2\% of our sample was diagnosed as risky; overall, doctors achieve 78.58\%
average accuracy. The doctors' general tendency to classify a large portion of
normal birth as risky is suggestive of a high cost of false negatives relative
to the false positives.

We exclude from the original data sample observations with missing information
on pregnancy outcome, as well as those for which more than 50\% of feature values
are missing. The final data used in the current analysis includes 192551 couples,
who are diagnosed by 1526 doctors.

\subsection{Machine Learning Algorithm}

We experiment with a number of machine learning binary classification algorithm
including logistic regression, deep neural networks, and regression trees. For
parametric models, deep neural networks can achieve good classification
performance. Similar to the simulation experiment, we bootstrap the whole sample
(100) times and draw the 95\% confidence bands based on the bootstrapped ROC
curves depicted in Figure \ref{roc for boot doctor}. We used
an 8:2 training set and test set partition. Note that there exist
categorical features in the data. For example, the mother's occupation is coded
as 1 for the farmer, 2 for the factory worker, etc. Neural networks are not good at
handling such structured data (\cite{guo2016entity}). To address this issue we
used an entities embedding method (\cite{guo2016entity}, \cite{yang2015embedding}) to map
these features to embedded vectors that can be directly utilized by a neural
network. We used a 5-layer fully-connected neural network for prediction, with
50, 50, 50, 50, 25 neurons per layer, and ReLU (Rectified Linear Unit) is used
as the activation function. A dropout method (\cite{srivastava2014dropout}) is
employed to provide regularization and in-model ensemble. Batch
Normalization (\cite{ioffe2015batch}) is used to accelerate network training and
Xavier's method (\cite{glorot2010understanding}) is used to initialization network
parameters.

\begin{figure}
\begin{center}
\caption{Bootstrapped ROC Curves for the Birth Defect Model.\small{\newline\newline Bold Blue line is for the original model, and other lines are bootstrapped ROC curves.}}\label{roc for boot doctor}
    \includegraphics[height=.32\textheight]{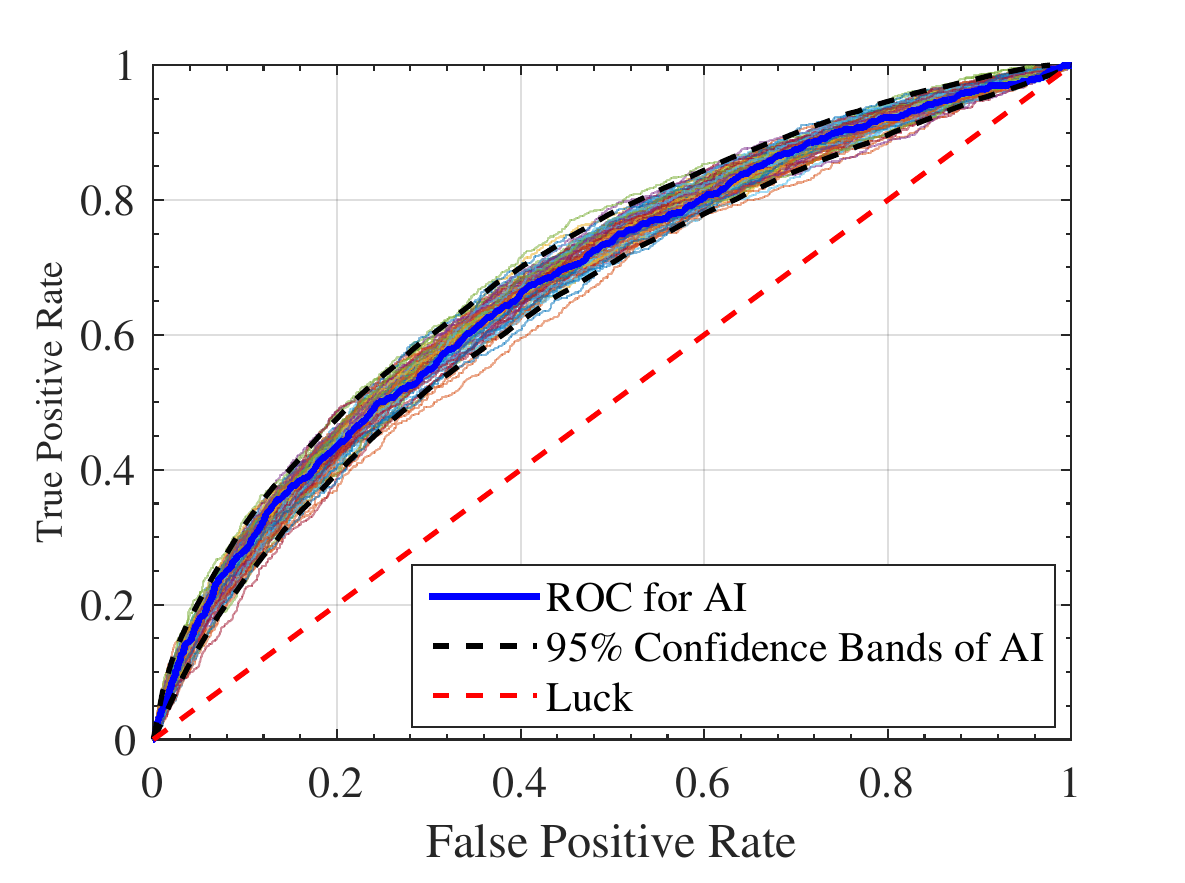}
\end{center}
\end{figure}

We then predict doctors' diagnoses, where the labels in both the training data and the test
data are the diagnosis by doctors, using the same deep learning algorithm as the one
used to fit the observed birth defect outcome. The difference is that the dependent variable
changes from birth defect outcomes to doctor's diagnoses. Figure \ref{roc for
predicting doctor} shows the ROC curve of this algorithm with an AUC of 0.828.
This procedure of training a model to predict the doctors' diagnosis and using
it to predict the same doctors' diagnosis is different from the PROC
curve discussed earlier where the model trained with the doctors' diagnosis is
used to predict the actual birth defect outcome.
The AUC of 0.828 for this ``predicted doctor's model'' is much higher than that
for both the PROC and the MROC. It comes with no surprise that it is much easier
to use features to predict the diagnosis by doctors, and to predict the actual
outcome of birth defects.

\begin{figure}
\begin{center}
\caption{Empirical ROC curve for the model predicting doctor's diagnoses}\label{roc for predicting doctor}
\includegraphics[height=.32\textheight]{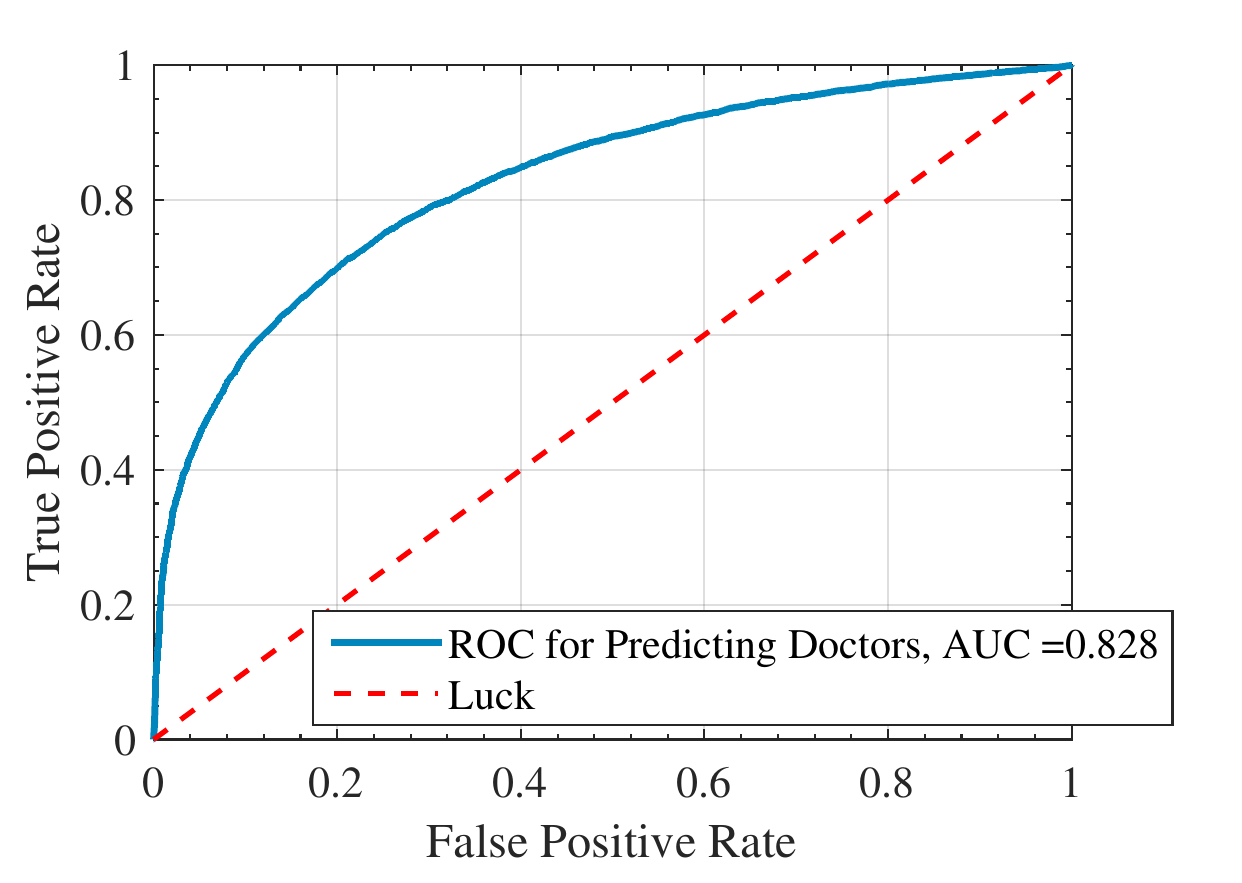}
\end{center}
\end{figure}

Figure \ref{roc for all} shows the MROC curve for the deep learning algorithm,
the PROC curve for the predicted doctor's model, and the FPR/TPR pair of average doctors.
The AUC for MROC of the the deep learning algorithm is 0.677. The aggregate
PTPR/PFPR pair of doctors lies significantly below the MROC. The PROC curve for
predicted doctor's model lies slightly above the aggregate doctors' TPR/FPR,
but below the machine learning ROC and has a lower AUC of about 0.533.
Note that it is possible for the predicted humans' model to outperform the
humans TPR/FPR pair, as evident in \cite{QJEbail}.
However, such an interpretation is subject to the concern due to Jensen's inequality.
Furthermore, the machine learning MROC and the ``predicted doctor'' PROC are two very different
curves that encode different types of information and are not directly comparable.

\begin{figure}
\begin{center}
\caption{Empirical ROC curve for machine learning, predicted doctors model and real aggravated doctors}\label{roc for all}
\includegraphics[height=.32\textheight]{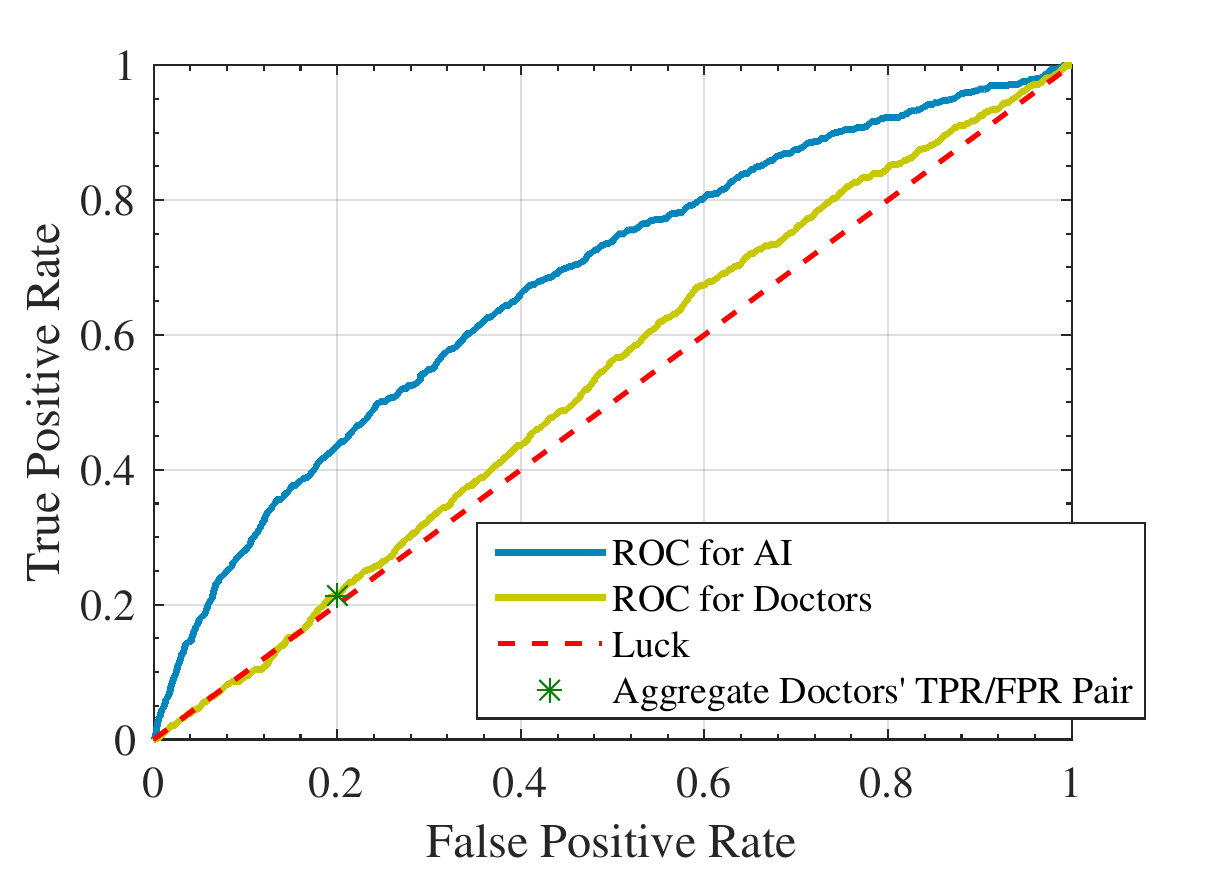}
\end{center}
\end{figure}

The $95\%$ confidence bands for the ROC curves of machines algorithm and
predicted doctors are plotted via a bootstrap method. Since our sample
size is large, the resulted bands are narrow. The AUCs of machine algorithm and
predicted doctors are, respectively, $0.677\pm0.012$ and $0.533\pm0.024$ ($95\%$
confidence bands). Given that the AUC of machine algorithm is much larger than
the predicted doctors, and both have small confidence bands. In the absence of
incentive heterogeneity in decision making, the data provides strong
statistical evidence in favor of the alternative hypothesis that the machine algorithm
is better than the predicted doctor model in forecasting the actual birth
defect outcomes.

If the doctors' diagnosis (which is observed in the data) are only based on
observed features $X$, it ought to be
learned by the machine algorithm with a high level of precision. Consequently, the ROC
for predicting doctors' diagnosis should be very close to the $\(0,1\)$ vertex
that represents perfect classification. However, as Figure \ref{roc for predicting doctor} shows, the ROC for predicting doctors is still far away from
perfect classification, in spite of the fact that doctors' decisions are much easier to learn
than the labeled actual outcomes.
This is evident of the presence of incentive or informational heterogeneity in
doctors' diagnosis beyond those encoded in the observed features $X$.

\subsection{Machines vs. Doctors}\label{Compare}

These findings from the data suggest provide potential evidence that both the
machine learning algorithm and the ``collective wisdom'' of doctors outperform
individual diagnosis, possibly because doctors misuse information.
However, such interpretations need to be viewed with much caution,
because essentially they
are based on the assumptions that
\begin{enumerate}
\item Doctors used a \emph{homogeneous} decision rule of the form of $\hat Y_i = \mathds{1}\(\hat q\(X_i\)>c\)$.
\item Doctors used the same or less information than the observed features $X$
to form an estimate of the propensity score $\hat q\(X_i\)$.
\end{enumerate}
where $c$ is a constant cutoff point employed by all doctors. Both of these are strong
assumptions that might be violated empirically. Possible alternative
explanations of these empirical results may be caused by the
\emph{heterogeneity} of decision rules, which might include:
\begin{enumerate}
	 \item Heterogeneous incentives (various cutoff values across doctors).
	 \item Cutoff values depending on observed features for individual doctors.
\end{enumerate}
We argue that incentive heterogeneity is likely to be present in an empirical data set. For example,
using a million observation electronic birth data from New Jersey, \cite{currie3}
find strong evidence of heterogeneous decision making by doctors.
\cite{lembke2012} suggests that online rating system tends to alter the
incentive of practicing doctors.
Some doctors care about online rating while some do not, creating incentive
heterogeneity across doctors.

On the one hand, the MROC from machine learning represents information encoded in the observed
features about the predictability of the outcome label. On the other hand, the aggregate TPR/FPR pair
in the data encodes not only hetergeneous information dispersed among decision makers
such as doctors, but also the heterogeneity in their incentives and preferences. Even if doctors
know the true propensity score and therefore have as much information
as the machine learning algorithm has, the aggregate PTPR/PFPR pair might still appear
to be inferior as long as doctors employ a loss function that
depends on the observed features. (Also refer to Section \ref{cutoff points}.) 
That the TPR/FPR pair lies below the MROC might be a mere reflection of preference heterogeneity among human decision
makers instead of inferior information, and by itself does not provide evidence against the
rationality and quality of human decision making.

\subsection{Incentive Heterogeneity}
We plotted the
individual FPR and TPR pairs for those doctors
who have diagnosed more than 100 patients in Figure \ref{figure RealDoctor100}.
It is evident that doctors' performance is highly
heterogeneous. Some doctors are particularly conservative, showing high FPR and TPR. Although
many doctors fall below the AI's ROC curve, a fraction of the doctors show very high TPR and
low FPR that lie above the machine learned ROC.
\begin{figure}
\begin{center}
\caption{The heterogeneity of Doctors}\label{figure RealDoctor100}
    \includegraphics[height=.32\textheight]{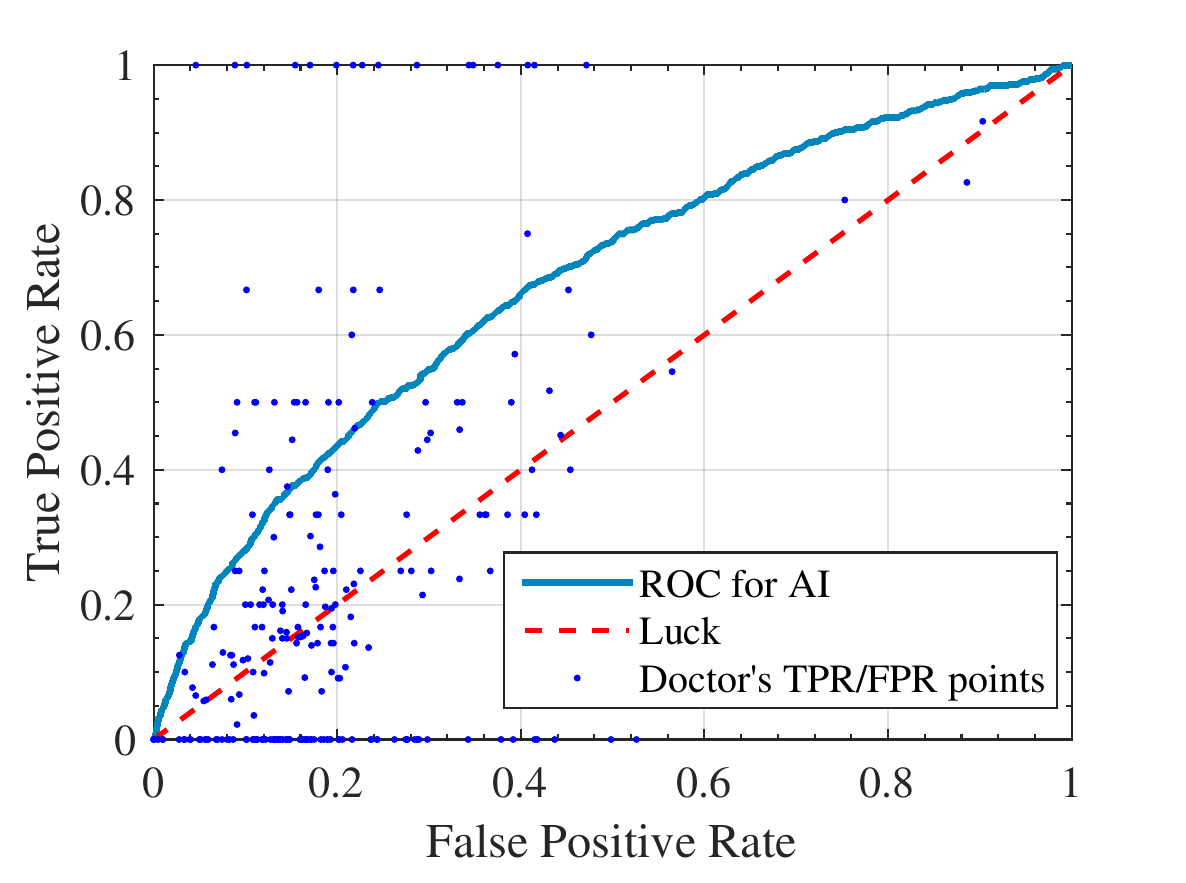}
\end{center}
\end{figure}

For each doctor, we obtain the machine propensity score $p\(x\)$ from machine
learning algorithm and doctors' propensity score $q\(x\)$ from the predicted
doctors' model.
We then run a regression using
polynomial regressors up to the power of 5 to obtain the relationship between
$p\(x\)$ and $q\(x\)$ (the $h$ function in section \ref{uncover}). To explicitly
account for heterogeneity in decision making, we examine two doctors, one capable
doctor with performance above the machine's ROC curve and one incapable doctor
with performance below the machine's ROC curve. These two doctors
see 1307 and 2254 patients, respectively, in our data set. Figure \ref{goodDoctor}
and \ref{badDoctor} show very different shapes for these two $h\(x\)$
functions, which is indicative of incentive heterogeneity among doctors.
\begin{figure}
\begin{center}
\caption{$p\(x\)$ vs. $q\(x\)$ for a capable doctor}\label{goodDoctor}
    \includegraphics[height=.32\textheight]{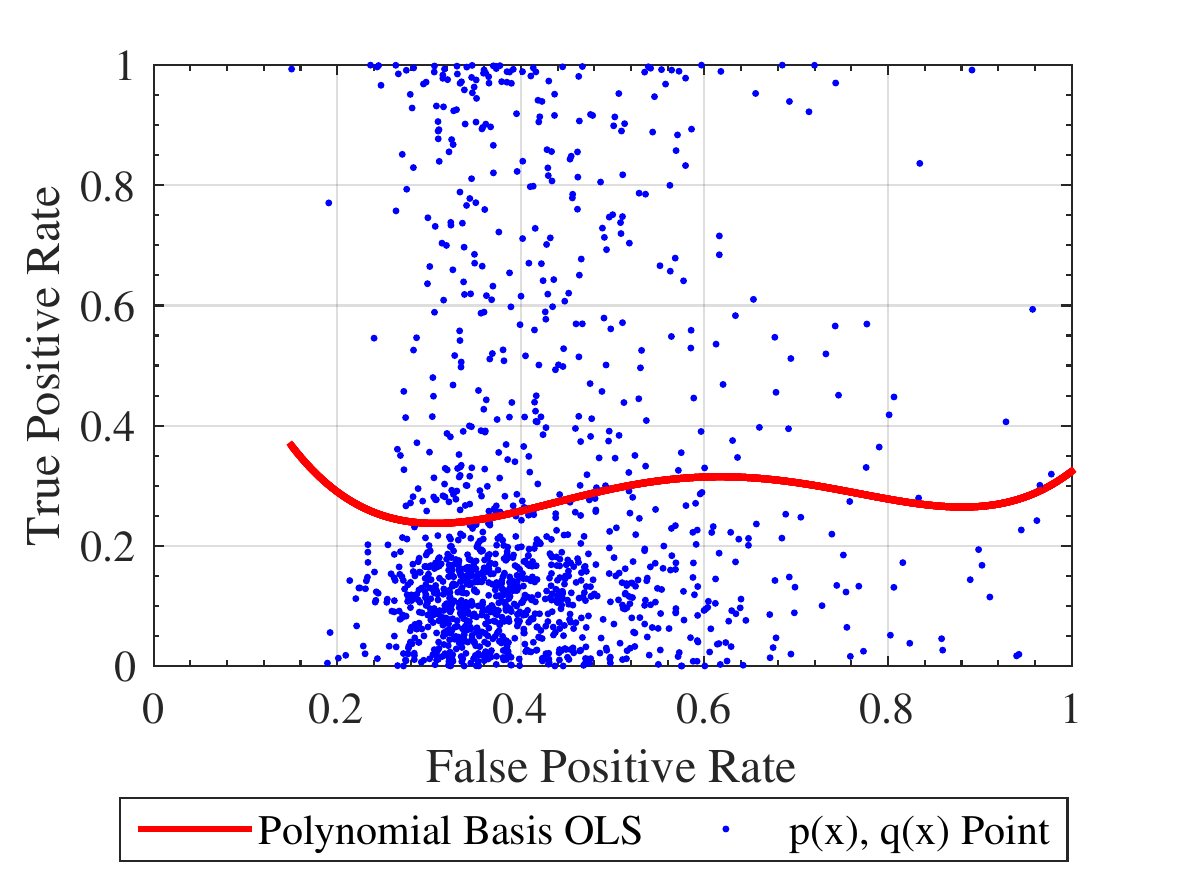}
\end{center}
\end{figure}

\begin{figure}
\begin{center}
\caption{$p\(x\)$ vs. $q\(x\)$ for an incapable doctor}\label{badDoctor}
    \includegraphics[height=.32\textheight]{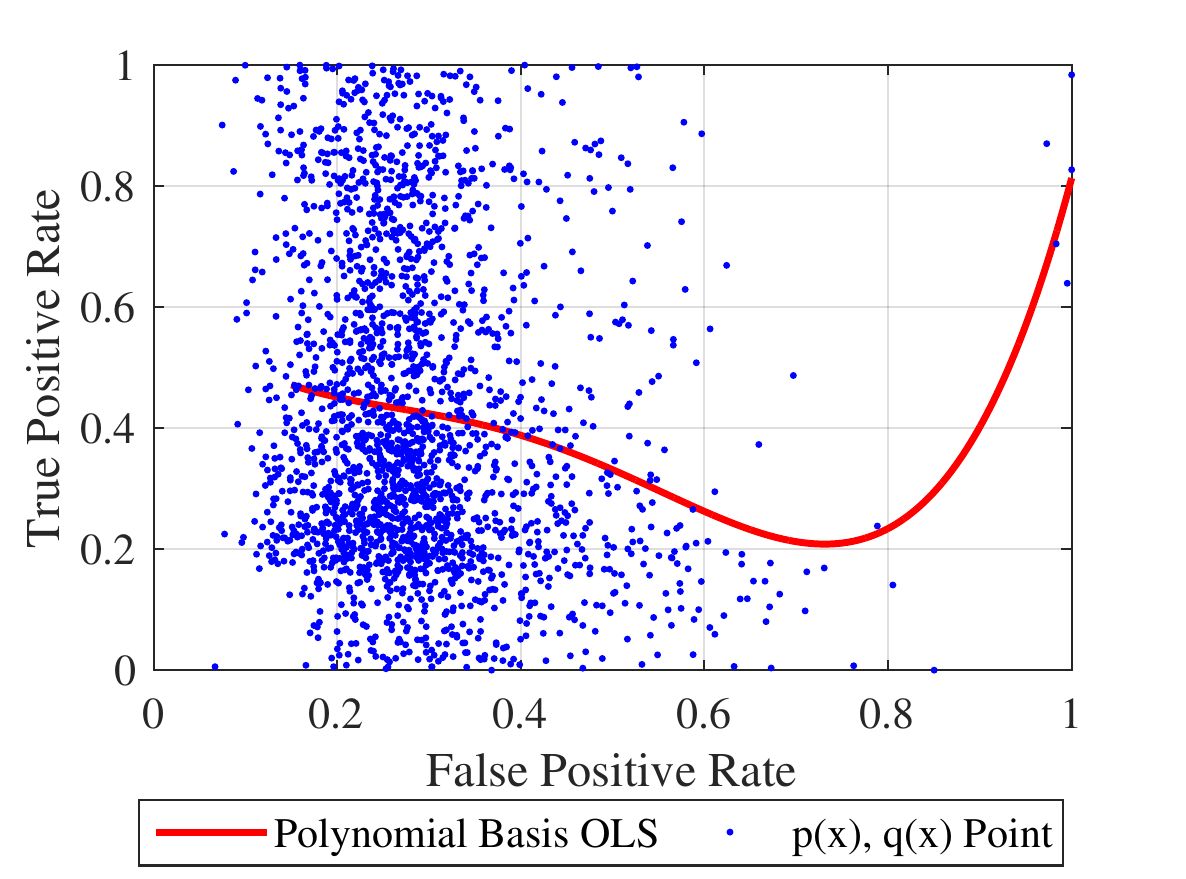}
\end{center}
\end{figure}
Furthermore, Figure \ref{goodDoctor} and \ref{badDoctor} show that neither
$h\(x\)$ function is monotonic, which implies the dependence of the cutoff
threshold on observed features (refer to section \ref{uncover}). We also obtain
the $h\(x\)$ function for all doctors, which is also a highly nonmonotonic
function. For brevity, we hence omit the corresponding figure.

To control for incentive heterogeneity we make use of the conditioning idea in
section \ref{cutoff points} when we
test the relative performance between humans and machines.
We use the characteristic called ``bad previous pregnancy outcomes'' as the conditioning
feature that the cutoff function depends on, where 0 stands for no bad previous outcomes and 1 for bad outcomes.
We first remove the set of doctors whose TPR/FPR pairs under perform AI (see Figure
\ref{roc for predicting doctor}).  There are 3597 observations of class 0 and 490 observations
of class 1 in the remaining data set.
We use machine learning to make predictions for these two samples separately.
The results are presented in the top two figures in Figure \ref{figure 14},
where the doctors' TPR/FPR pairs are on (slightly above) the corresponding AI's
ROC curves.  Then, we use AI to predict the combined data set. This time,
the doctors' point falls below the AI's ROC curve.
\begin{figure}
\begin{center}
\caption{Empirical ROC curves and doctors' performance conditional on the feature of "previous pregnancy outcomes"}\label{figure 14}
\subfigure{\includegraphics[height=.25\textheight]{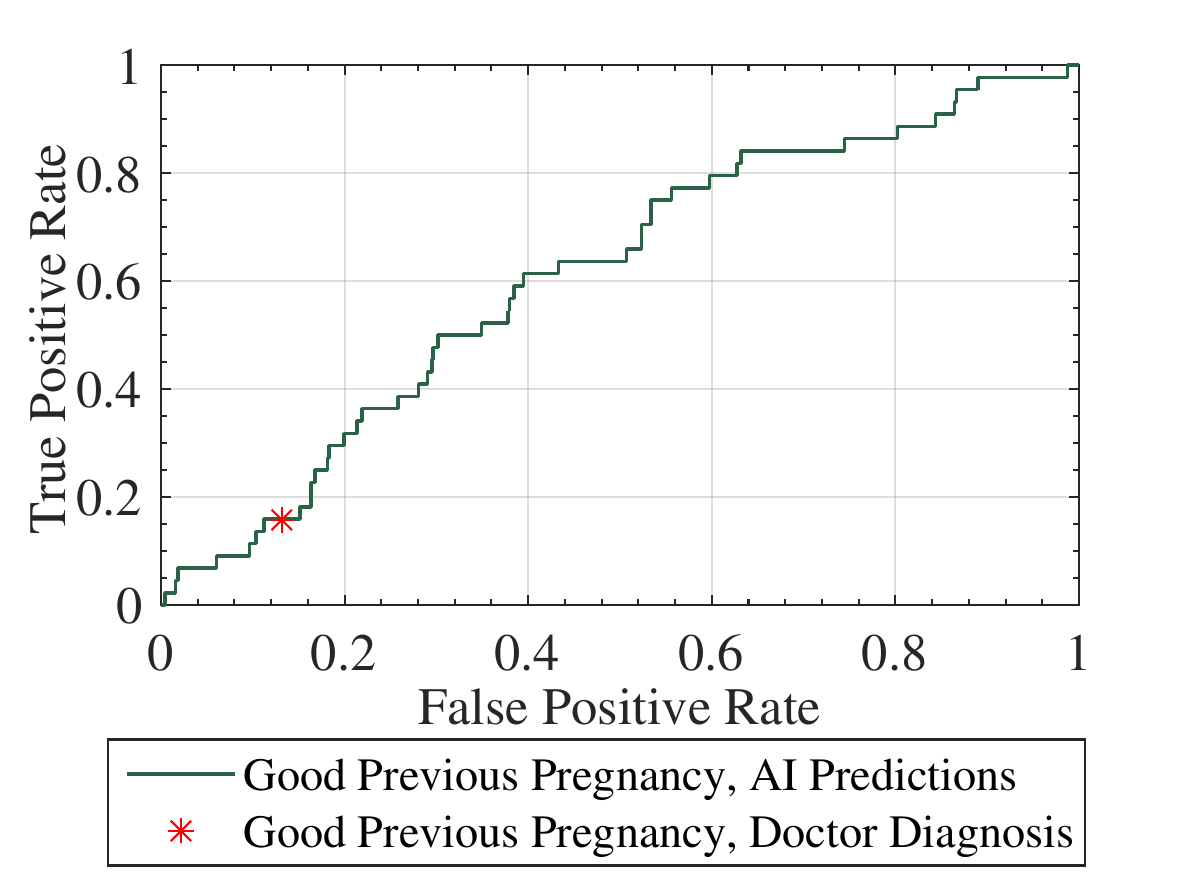}}
	\subfigure{\includegraphics[height=.25\textheight]{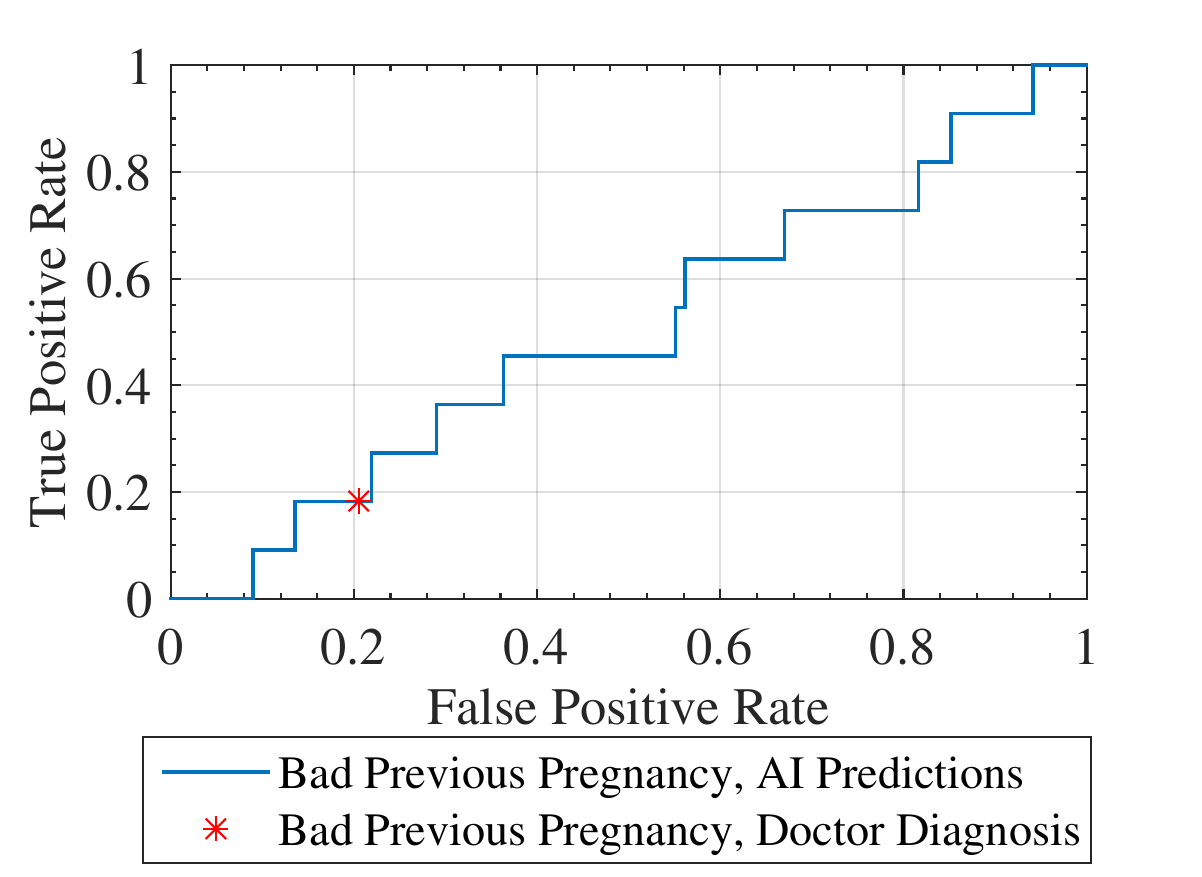}}\\
	\subfigure{\includegraphics[height=.25\textheight]{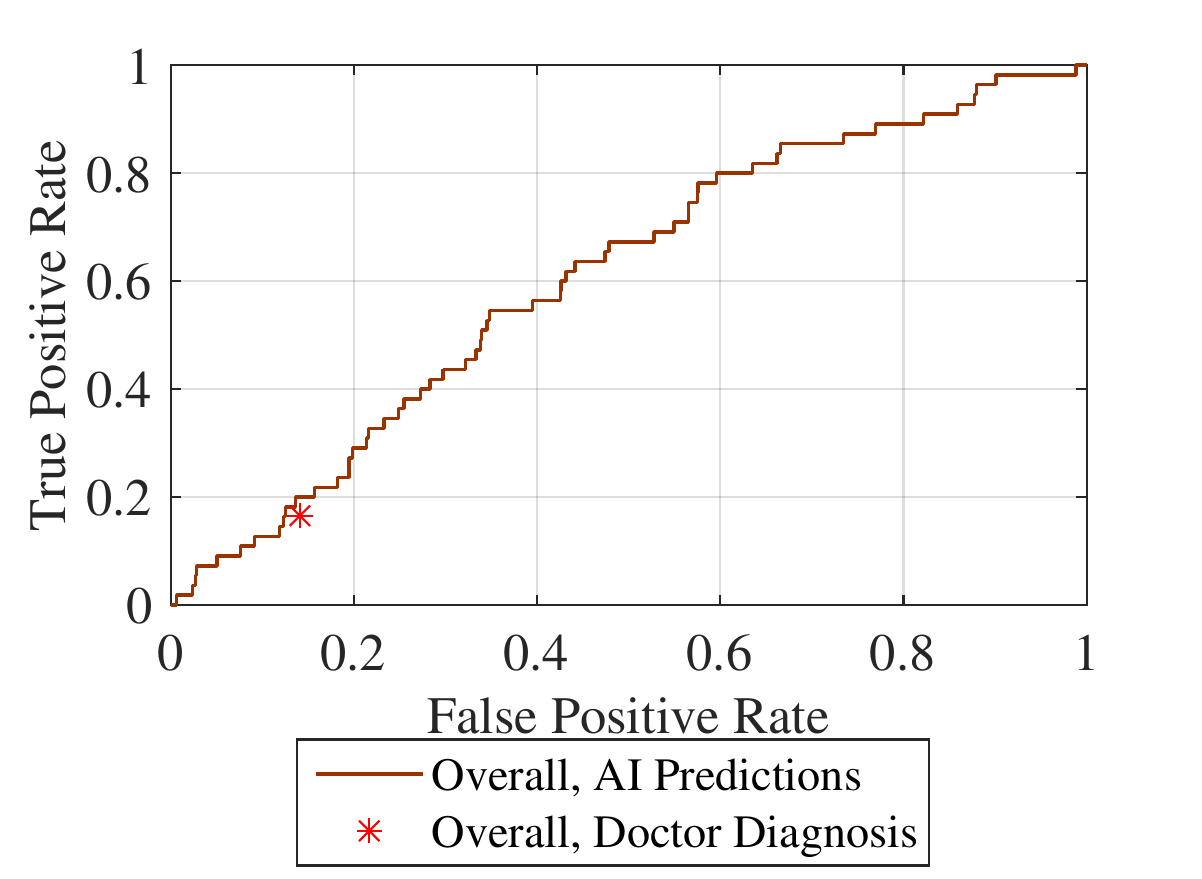}}
\end{center}
\end{figure}

As discussed in Section \ref{cutoff points}, Figure \ref{figure 14} suggests that doctors' cutoff
thresholds depend on ``bad previous pregnancy outcomes''. It also shows that even though
doctors are better than AI, their aggregate TPR/FPR pair can be below AI's ROC curve due
to incentive heterogeneity with the dependence of cutoff points on features.

\section{Conclusion}
This paper studies the statistical properties of the
Receiver Operating Characteristic (ROC) curve and its relation to binary classification and
human decision making. We derive results that are useful for the construction of
confidence bands for the ROC curve.
We also show that a maximum AUC estimator is essentially the maximum
score estimator from the econometrics literature, and discuss its
implications for model selection and testing in binary classification
problems. We use the Pre-Pregnancy Checkups of
reproductive age couples in Henan Province provided by the Chinese Ministry of
Health as an illustrative example for our theoretical model on ROC curves.

The advent of machine learning and artificial
intelligence offers the potential of improving human decision making. A natural
question to ask is whether machine learning can replace human decision makers.
We conclude that the most commonly used performance measure, ROC curves, generated by a machine learning algorithm relative to
a human decision maker does not translate into a statement about the superiority
of the machine learning algorithm. We provide evidences in which
human decision makers are fully rational yet appear to lie below the machine
learned ROC curve. This is because ROC curve only provides a type of summary of
information, yet decision making involves both information and incentives.
These findings serve to clarify the misconception about ROC curves that appear in
previous research papers.

\section{Acknowledgements}
We thank Xiaohong Chen, Andres Santos, 
and participants at various seminars and conferences for insightful comments.
This study was approved by the National Health and Family Planning Commission. Informed consents were obtained from all the NFPC participants. We acknowledge funding support from the National Science Foundation (SES 1658950 to Han Hong), the National Science Fund for Distinguished Young Scholars of China (71325007 to Ke Tang), the State's Key Project of Research and Development Plan (2016YFC1000307 to Jingyuan Wang) and the National Natural Science Foundation of China (61572059 to Jingyuan Wang). 

\appendix
\section{Appendix}

\begin{lemma}\label{concavity of roc lemma}
Denote $f_p\(u\)$ the implied
density of the propensity score $p\(x\)$ induced by $f_X\(x\)$. Let
$f_p\(u\)$ be positive and continuous on $u\in \[0,1\]$, then the ROC
curve corresponding to a correctly specified $p\(x\)$ is a concave function.
\end{lemma}

\begin{proof}
We can write, where both size and power are now indexed by $c$:
 {\begin{align}\begin{split}\nonumber
	\beta\(c\) =& \frac{1}{p}\int \mathds{1}\(p\(X\) > c\) p\(X\) f\(X\) \mathrm{d}X = \frac{1}{p}\int_c^\infty u
	f_P\(u\) \mathrm{d}u
 \end{split}\end{align}}
and
 {\begin{align}\begin{split}\nonumber
	\alpha\(c\)=\frac{1}{1 - p}\int \mathds{1}\(p\(X\) > c\) \(1-p\(X\)\) f\(X\)\mathrm{d}X =
	\frac{1}{1 - p}\int_c^\infty \(1-u\)
	f_p\(u\) \mathrm{d}u.
 \end{split}\end{align}}
Hence
 {\begin{align}\begin{split}\nonumber
	\frac{\mathrm{d}\beta}{\mathrm{d} c} =& -c \frac{f_p\(c\)}{p}, \quad \frac{\mathrm{d}\alpha}{\mathrm{d}c} =
	-\(1-c\) \frac{f_p\(c\)}{1-p},\quad
	 \frac{\mathrm{d}\beta}{\mathrm{d}\alpha} = \frac{c}{1-c},
 \end{split}\end{align}}
which is strictly increasing in $c$, and decreases when $c$
decreases from $1$ (at the origin of $\(\alpha,\beta\)=\(0,0\)$) to $0$ (at the
$\(\alpha,\beta\)=(1,1)$ vertex).
\end{proof}

\begin{lemma}\label{ROCHetero}
In the absence of information heterogeneity, if the cutoff threshold $c$ varies among decision
makers or within a single decision maker (incentive heterogeneity), the aggregate
PTPR/PFPR pair of decision makers is below the optimal ROC curve.
\end{lemma}

\begin{proof}
To prove this, we write the pair as
{\begin{align}\begin{split}\nonumber
\text{PFPR} =&\frac{1}{1-p} \int \lambda\(x\) \(1-p\(x\)\) f\(x\) \mathrm{d}x,\\
\text{PTPR} =&\frac{1}{p} \int \lambda\(x\) p\(x\) f\(x\) \mathrm{d}x,
\end{split}\end{align}
}
where $\lambda\(x\)=\int \mathds{1}\(p\(x\)>h\(x,v\)\)f\(v\)\mathrm{d}v$, note that the cutoff variable $h\(x,v\)$ depends on $x$ and a random variable $v$. We assume that the decision rule $\hat Y=\mathds{1}\(p\(x\)>h\(x,v\)\)$ has some classification ability, i.e. $\mathds{1}\(p\(x\)>h\(x,v\)\)\not\equiv 0$ and $\mathds{1}\(p\(x\)>h\(x,v\)\)\not\equiv 1$, therefore, $0<\lambda\(x\)<1$.
Then for a $c^*$ satisfying
		 {\begin{align}
			 \alpha\(c^*\) =& \int \mathds{1}\(p\(x\) > c^*\) \(1-p\(x\)\) f\(x\) \mathrm{d}x
 =
\text{PFPR} = \int \lambda\(x\) \(1-p\(x\)\) f\(x\) \mathrm{d}x
		 \end{align}
		 }
on the optimal ROC curve, given the FPR $\alpha\(c^*\)$, we can find the TPR $\beta\(c^*\)$:
{\begin{align}
			 \beta\(c^*\) =& \int \mathds{1}\(p\(x\) > c^*\) p\(x\) f\(x\) \mathrm{d}x.
	 \end{align}}
Since $\alpha\(c^*\)$ and $\beta\(c^*\)$ are on the optimal ROC curve, by the Neyman-Pearson Lemma,
there must be some positive $\eta_1$ and $\eta_2$ such that $\mathds{1}\(p\(x\) > c^*\)$ (but not $\lambda\(x\)$) solves
{\begin{align}
\arg\max_{\phi\(\cdot\)}
\eta_1  \int \phi\(x\) p\(x\) f\(x\) \mathrm{d}x
- \eta_2  \int \phi\(x\) \(1-p\(x\)\) f\(x\) \mathrm{d}x.
\end{align}
}
Hence $\eta_1 \beta\(c^*\) - \eta_2 \alpha\(c^*\) > \eta_1 \text{PTPR}
- \eta_2 \text{PFPR}$, and thus $\beta\(c^*\) > \text{PTPR}$.
\end{proof}

\begin{lemma}\label{same ROC if and only if}
Let $p\(X\)$ and $q\(X\)$ be continuously distributed with a positive density on $\[0,1\]$,
and $p\(X\)$ be correctly specified. These two propensity score
functions $p\(X\)$ and $q\(X\)$ correspond to the same
ROC if and if only if one is a strict monotonic transformation of the other.
\end{lemma}

\begin{proof}
The {\it if} part is immediate. Any point on the ROC of $p\(X\)$ maps into
another point on th ROC of $q\(X\)$. Now consider the {\it only if} part.
First we show that there exists a function $\bar c\(t\)$, such that for all $t$
\begin{align}\begin{split}\nonumber
\mathds{1}\(q\(X\) \geq t\) = \mathds{1}\(p\(X\) \geq \bar c\(t\)\)
\end{split}\end{align}
By assumption of the identical ROCs, $\bar c\(t\)$ exists such that,
\begin{align}\begin{split}\nonumber
\int \mathds{1}\(q\(x\) \geq t\) p\(x\) f\(x\) dx =&
\int \mathds{1}\(p\(x\) \geq \bar c\(t\)\) p\(x\) f\(x\) dx, \\
\int \mathds{1}\(q\(x\) \geq t\) \(1-p\(x\)\) f\(x\) dx =&
\int \mathds{1}\(p\(x\) \geq \bar c\(t\)\) \(1-p\(x\)\) f\(x\) dx.
\end{split}\end{align}
Note that the left hand sides are decreasing in $t$, and the right hand sides
are decreasing in $\bar c\(t\)$, $\bar c\(t\)$ is necessarily an increasing
function of $t$.
To simplify, for $f\(u,z\)$ denoting the implied joint density of $\(p\(X\),
q\(X\)\)$, we can write
\begin{align}\begin{split}\nonumber
\iint \mathds{1}\(z \geq t\) u f\(u,z\) \mathrm{d}u \mathrm{d}z  =& \iint
\mathds{1}\(u \geq \bar c\(t\)\) u f\(u,z\) \mathrm{d}u\\
\iint \mathds{1}\(z \geq t\)  f\(u,z\) \mathrm{d}u \mathrm{d}z  =& \iint
\mathds{1}\(u \geq \bar c\(t\)\)  f\(u,z\) \mathrm{d}u
\mathrm{d}z
\end{split}\end{align}
which can be written as
\begin{align}\begin{split}\label{prepare for neyman-pearson argument}
\int \mathbb{E}\[\mathds{1}\(z \geq t\) \vert u\] u f\(u\) \mathrm{d}u =& \int
\mathds{1}\(u \geq \bar c\(t\)\) u f\(u\) \mathrm{d}u\\
\int \mathbb{E}\[\mathds{1}\(z \geq t\) \vert u\] \(1-u\) f\(u\) \mathrm{d}u =& \int
\mathds{1}\(u \geq \bar c\(t\)\) \(1-u\) f\(u\) \mathrm{d}u.
\end{split}\end{align}
By Neyman-Pearson Lemma's argument,
\begin{align}\begin{split}\label{z and u are the same}
\mathbb{E}\[\mathds{1}\(z \geq t\) \vert u\] = \mathds{1}\(u\geq \bar c\(t\)\)
\end{split}\end{align}
which implies that $\mathds{1}\(z \geq t\) = \mathds{1}\(u\geq \bar c\(t\)\)$.
To see \eqref{z and u are the same}, let $h\(u,t\) \equiv
\mathbb{E}\[\mathds{1}\(z \geq t\) \vert u\]$. Take a linear combination of
the two equalities in \eqref{prepare for neyman-pearson argument} using
$\(1-\bar c\(t\)\)$ and $-\bar c\(t\)$,
\begin{align}\begin{split}\nonumber
\int h\(u,t\) \[
\(1-\bar c\(t\)\) u - \bar c\(t\) \(1-u\)
\] f\(u\) du.
\end{split}\end{align}
The function $h\(u,t\)$ that maximizes the above integral is obviously
$h\(u,t\) = 1\(u \geq c\(t\)\)$, which appears on the right hand sides of
 \eqref{prepare for neyman-pearson argument}. Therefore in order for both
inequalities in \eqref{prepare for neyman-pearson argument} to hold, it must be
the case that $h\(u,t\)\equiv \mathbb{E}\[\mathds{1}\(z \geq t\) \vert u\] =
1\(u \geq c\(t\)\)$.

Next we verify again that $\bar c\(t\)$ is strictly monotonic.
Define $X_p\(c\) = \{x: p\(x\) \geq c\}$, and $X_q\(t\) = \{x: q\(x\) \geq t\}$.
By definition,
\begin{align}\begin{split}\nonumber
X_p\(c_2\)  \subset X_p\(c_1\) \quad \text{iff}\quad c_1 < c_2,\quad
X_q\(t_2\)  \subset X_q\(t_1\) \quad \text{iff}\quad t_1 < t_2.
\end{split}\end{align}
Then for $X_q\(t\)  = X_p\(\bar c\(t\)\)$, it must be that $\bar c\(t_1\) > \bar
c\(t_2\)$ if $t_1 > t_2$. Then an inverse function $\bar c^{-1}\(\cdot\)$
exists, and for all $t$,
\begin{align}\begin{split}\nonumber
X_q\(t\) = X_p\(\bar c\(t\)\) = \{x: \bar c^{-1}\(p\(x\)\) \geq t\}.
\end{split}\end{align}
So that $q\(x\) = \bar c^{-1}\(p\(x\)\)$.
\end{proof}

\begin{lemma}\label{ShrinkSingleton}
In the absence of information and incentive heterogeneity beyond the observed
features, the predicted doctor ROC degenerates to a singleton.
\end{lemma}

\begin{proof}
To see this, the predicted human is a step function that takes only values $0$ and $1$:
if $D = 1\(\bar p\(X\) \geq c_0\)$, then
{\begin{align}\begin{split}\nonumber
q\(X\)\equiv \mathbb{P}\(D=1\vert X\) =\mathbb{P}\(\bar p\(X\) > c_0 \vert X\)= \mathds{1}\(\bar p\(X\) > c_0\),
\end{split}\end{align}}
The resulting ROC curve for the predicted human is given by
{
\begin{align}\begin{split}\nonumber
\bar\beta\(c\)
=& \frac1p \int \mathds{1}\(1\(\bar p\(x\) > c_0\) > c\) p\(x\) f\(x\) \mathrm{d}x \\
\bar\alpha\(c\)
=& \frac{1}{1-p} \int \mathds{1}\(1\(\bar p\(x\) > c_0\) > c\) \(1-p\(x\)\) f\(x\) \mathrm{d}x.
\end{split}\end{align}
}
		 Note that 
{\begin{align}\begin{split}\nonumber
	\bar\beta\(c\)
	=& \frac1p \int \mathds{1}\(\bar p\(x\) > c_0\) p\(x\) f\(x\) \mathrm{d}x,\\
	\bar\alpha\(c\) =&
\frac{1}{1-p} \int \mathds{1}\(\bar p\(x\) > c_0\) \(1-p\(x\)\) f\(x\) \mathrm{d}x.
\end{split}\end{align}
}
which does not even depend on $c$. Therefore the ROC of the predicted doctor
		 should be a singleton point.
\end{proof}

\begin{lemma}\label{aboveROC}
If humans make correct use of additional information $U$ beyond the observed
features $X$, and predicted humans' ROC lies above the machine learned ROC.
\end{lemma}

\begin{proof}
For a formal proof, note that the ROC curve is defined by the locus of the pair of
functions
$\(\alpha\(c\), \beta\(c\), c\in \(0,1\)\),$
	where
{
\begin{align}\begin{split}\nonumber
\alpha\(c\) =& \frac{1}{1 - p}\iint \mathds{1}\(p\(x,u\) > c\) \(1-p\(x,u\)\) f\(x,u\) \mathrm{d}x \mathrm{d}u\\
\beta\(c\) =& \frac{1}{p}\iint \mathds{1}\(p\(x,u\) > c\) p\(x,u\) f\(x,u\)
\mathrm{d}x \mathrm{d}u.
\end{split}\end{align}
}
By definition and by Neyman-Pearson arguments, it must lie above the ROC curve defined by
the locus of
$\(\bar\alpha\(c\), \bar\beta\(c\), c\in \(0,1\)\)$, where
{
\begin{align}\begin{split}\nonumber
	\bar\alpha\(c\) =& \frac{1}{1-p} \iint \mathds{1}\(p\(x\) > c\) \(1-p\(x,u\)\) f\(x,u\) \mathrm{d}x \mathrm{d}u  \\
\bar\beta\(c\) =& \frac1p \iint \mathds{1}\(p\(x\) > c\) p\(x,u\) f\(x,u\) \mathrm{d}x \mathrm{d}u.
\end{split}\end{align}
}
But
{
\begin{align}\begin{split}\nonumber
\bar\beta\(c\) = \frac1p \int \mathds{1}\(p\(x\) > c\) \int p\(x,u\) f\(u \vert x\) \mathrm{d}u
f\(x\) \mathrm{d}x
= \frac1p \int \mathds{1}\(p\(x\) > c\) p\(x\) f\(x\) \mathrm{d}x,
\end{split}\end{align}
}
and similarly,
{
\begin{align}\begin{split}\nonumber
	\bar\alpha\(c\) =& \frac{1}{1-p} \int \mathds{1}\(p\(x\) > c\) \int \(1-p\(x,u\)\) f\(u \vert x\) \mathrm{d}u
f\(x\) \mathrm{d}x  \\
	=& \frac1p \int \mathds{1}\(p\(x\) > c\) \(1-p\(x\)\) f\(x\) \mathrm{d}x,
\end{split}\end{align}
}
Consequently, $\bar\alpha\(c\), \bar\beta\(c\), c \in \(0,1\)$ is
also the ROC curve defined by $p\(X\)$.

Therefore, for the case of homogeneous preference but heterogeneity in information, if the information is correct, the aggregate PTPR/PFPR pair must lie above the machine optimal ROC.
\end{proof}

\paragraph{Generalization of \cite{sherman1993limiting} under misspecification}

Without the single index assumption, suitable location and scale normalization, including removing
the constant term in $X_i$ and fixing the first element of $\theta$ at $1$ as in
as in \cite{sherman1993limiting}, is still necessary for identification. Here we
describe the generalization of section 6 in \cite{sherman1993limiting} to allow
for misspecification of the single index model.

Recall the kernel function for the U-statistics on pp 129 \cite{sherman1993limiting}, for $Z=\(Y,X\)$,
\begin{align}\begin{split}\nonumber
\tau\(z,\theta\)=&\mathbb{E} \[\mathds{1}\(x'\beta\(\theta\) > X'\beta\(\theta\)\) \mathds{1}\(y > Y\)
+\mathds{1}\(x'\beta\(\theta\) < X'\theta\) \mathds{1}\(y < Y\)\]\\
=&
\mathbb{E} \[\mathds{1}\(x'\beta\(\theta\) > X_j'\beta\(\theta\)\) \(\mathds{1}\(y > Y_j\)
- \mathds{1}\(y < Y_j\)\)\] + \mathbb{E} \[\mathds{1}\(y < Y_j\)\]
 \end{split}\end{align}
Without requiring the single index model, Theorem 4 of
\cite{sherman1993limiting} shows that
\begin{align}\begin{split}\nonumber
\sqrt{n}\(\hat\theta-\theta^*\)
= \frac{1}{\sqrt{n}} \sum_{i=1}^n  \kappa_i + o_P\(\frac{1}{\sqrt{n}}\),
\quad \kappa_i= V^{-1} \frac{\partial}{\partial \theta}
\tau\(Z_i, \theta\),  \quad
2V = \frac{\partial^2}{\partial \theta\partial \theta'} E \tau\(Z_i, \theta\).
\end{split}\end{align}
The influence function can be computed without referencing the single index
assumption and through the use of a Dirac function $\delta_0\(\cdot\)$,
\begin{align}
    \begin{split}\nonumber
\frac{\partial}{\partial \theta} \tau\(Z_i, \theta\)
        =&
        \frac{\partial}{\partial\theta} \iint
        \mathds{1}\(\(X_i - w\)'\beta\(\theta\) > 0\) \(\mathds{1}\(Y_i > z\) - \mathds{1}\(Y_i < z\)\) f_{Y, X}\(z, w\) \mathrm{d}w
        \mathrm{d}z\\
        =&
        \int
        \(\bar X_i -\bar w\) \delta_0\(\(X_i - w\)'\beta\(\theta\)\) S\(Y_i, w\)  f_X\(w\) \mathrm{d}w
    \end{split}
\end{align}
where $\bar X_i$ and $\bar w$ denote the corresponding vectors without the first
element,
\begin{align}\begin{split}\nonumber
S\(Y_i, w\) = \int \(\mathds{1}\(Y_i > z\) - \mathds{1}\(Y_i < z\)\)  f_{Y \vert X}\(z
\vert w\) \mathrm{d}z
= 2 F_{Y \vert X}\(Y_i \vert w\) - 1.
 \end{split}\end{align}
Consider the transformation $w \Longrightarrow  \(t = w'\beta\(\theta\)=w_1 + \bar
w'\theta, \bar w\)$, with the inverse transformation
$\(t, \bar w\) \Longrightarrow \(t - \theta'\bar w, \bar w\)$.
Note that the determinant of this transform is $1$. Then we can apply this transformation to
write
\begin{align}
    \begin{split}\nonumber
\frac{\partial}{\partial \theta} \tau\(Z_i, \theta\)
=&
        \iint
        \(\bar X_i -\bar w\) \delta_0\(t-X_i'\beta\(\theta\)\) S\(Y_i, \(t-\bar w'\theta,\bar w\)\)
        f_X\(\bar w, t\) \mathrm{d}\bar w \mathrm{d}t\\
        =&\int
        \(\bar X_i -\bar w\) S\(Y_i, \(X_i'\beta\(\theta\)-\bar w'\theta,\bar w\)\)
        f_X\(\bar w, X_i'\beta\(\theta\)\) \mathrm{d}\bar w
    \end{split}
\end{align}
Further ``simplication'' beyond this point does not appear to be feasible
without the single index and correct specification conditions in \cite{sherman1993limiting}.
Yet the above influence function is still valid for asymptotic inference.
The Hessian term can be computed by further differentiation,
\begin{align}
    \begin{split}\nonumber
        2V =\iint \frac{\partial}{\partial\theta}
\frac{\partial}{\partial \theta} \tau\(z, \theta\)
f_{Y,X}\(y, x\)
        \mathrm{d}y\mathrm{d}x
    \end{split}
\end{align}
which in general will
involve both $\frac{\partial}{\partial\theta}
S\(Y_i, \(t-\bar w'\theta,\bar w\)\)$ and
$\frac{\partial}{\partial\theta} f_X\(\bar w, X_i'\beta\(\theta\)\)$.

Similar to \cite{sherman1993limiting}, numerical differentiation can be employed to consistently
estimate the influence function (hence its asymptotic variance as well), and the
limiting Hessian, using step sizes that satisfy $\sqrt{n} \epsilon_n \rightarrow
\infty$ and $n^{1/4} \epsilon_n \rightarrow \infty$, respectively.
The arguments for the consistency of numerical derivatives
in \cite{sherman1993limiting} do not depend on a correct specified single index model
either, and only depend on the Euclidean properties of a monotone transformation
of the linear index function. See also \cite{hong2015extremum}.

\FloatBarrier
\phantomsection
\addcontentsline{toc}{section}{References}
\ifx\undefined\bysame
\newcommand{\bysame}{\leavevmode\hbox to\leftmargin{\hrulefill\,\,}}
\fi

\bibliographystyle{aer}

\end{document}